\documentclass[journal]{IEEEtran}

\usepackage{amsfonts,amsmath,amsthm,bm}
\usepackage{cite}
\usepackage{amsmath,amssymb,amsfonts}
\usepackage{algorithmic}
\usepackage{graphicx}
\usepackage{textcomp}
\usepackage{xcolor}
\usepackage[T1]{fontenc}
\usepackage{comment}
\usepackage{eucal}
\usepackage{algorithm}
\usepackage{algorithmic}
\usepackage[acronyms,nonumberlist,nopostdot,nomain,nogroupskip]{glossaries}

\usepackage{subcaption}
\usepackage[font=footnotesize]{caption}
\usepackage{booktabs}
\usepackage{hyperref}
\hypersetup{
    colorlinks=true,        
    linkcolor=blue,         
    citecolor=blue,          
    urlcolor=black
}
\usepackage{orcidlink}

\newtheorem{lem}{Lemma}

\newtheorem{remk}{Remark}

\newtheorem{prop}{Proposition}

\DeclareMathOperator{\sign}{sign}
\newcommand{\defeq}{\stackrel{\triangle}{=}}
\def\bOmega{\boldsymbol{\Omega}}
\newcommand{\e}{\mathsf{e}}
\newcommand{\jj}{\mathsf{j}}
\newcommand{\Herm}{\mathsf{H}}
\newcommand{\Trans}{\mathsf{T}}
\newcommand{\x}{\mathsf{x}}
\newcommand{\y}{\mathsf{y}}
\newcommand{\z}{\mathsf{z}}
\newcommand{\bx}{\mathbf{x}}

\newcommand{\bW}{\mathbf{W}}
\newcommand{\bs}{\mathbf{s}}
\newcommand{\bn}{\mathbf{n}}
\newcommand{\bH}{\mathbf{H}}
\newcommand{\ba}{\mathbf{a}}

\newcommand{\bu}{\mathbf{u}}

\newcommand{\br}{\mathbf{r}}
\newcommand{\bbm}{\mathbf{m}}
\newcommand{\bbM}{\mathbf{M}}

\newcommand{\bI}{\mathbf{I}}
\newcommand{\bX}{\mathbf{X}}
\newcommand{\bh}{\mathbf{h}}

\newcommand{\bq}{\mathbf{q}}
\newcommand{\bv}{\mathbf{v}}
\newcommand{\bt}{\mathbf{t}}
\newcommand{\bE}{\mathbf{E}}

\newcommand{\kk}{\kappa}
\newcommand{\bDelta}{\mathbf{\Delta}}

\newcommand{\bGamma}{\boldsymbol{\Gamma}}
\newcommand{\btau}{\boldsymbol{\tau}}
\newcommand{\blambda}{\boldsymbol{\lambda}}
\newcommand{\bLambda}{\boldsymbol{\Lambda}}
\newcommand{\bxi}{\boldsymbol{\xi}}
\newcommand{\bXi}{\boldsymbol{\Xi}}

\newcommand{\Ex}{\mathbb{E}}
\newcommand{\Var}{\mathbb{V}}
\newcommand{\diag}{\mathrm{diag}}

\newcommand{\dd}{\mathrm{d}}

\newcommand{\tx}{\mathrm{tx}}
\newcommand{\thr}{\mathrm{thr}}
\newcommand{\SNR}{\mathrm{SNR}}
\newcommand{\RF}{\mathrm{RF}}

\newcommand{\lc}{\mathrm{lc}}
\newcommand{\elastic}{\mathrm{elastic}}
\newcommand{\electric}{\mathrm{electric}}

\newcommand{\doubt}{\textcolor{red}{??? }}

\def\bphi{\boldsymbol{\phi}}
\def\bDelta{\boldsymbol{\Delta}}

\def\bomega{\boldsymbol{\omega}}

\def\bzero{\boldsymbol{0}}
\def\bone{\boldsymbol{1}}
\def\bvarepsilon{\boldsymbol{\varepsilon}}
\def\Cset{\mathbb{C}}

\def\Rset{\mathbb{R}}
\def\Zset{\mathbb{Z}}
\def\LOS{\mathrm{LOS}}
\def\nLOS{\mathrm{nLOS}}
\def\tmax{\mathrm{max}}
\def\tx{\mathrm{tx}}
\def\rx{\mathrm{rx}}
\def\BS{\mathrm{BS}}
\def\RIS{\mathrm{RIS}}

\def\SNR{\mathrm{SNR}}
\def\mol{\mathrm{mol}}

\def\thr{\mathrm{thr}}

\def\sCN{\mathcal{CN}}

\def\Wset{\mathcal{W}}
\def\Uset{\mathcal{U}}

\def\bigO{\mathcal{O}}
\usepackage{xcolor}
\definecolor{mygreen}{HTML}{4CBB17}

\newcommand{\MDD}[1]{{\color{black} #1}}

\newcommand{\MD}[1]{{\color{black} #1}}

\newcommand\copyrighttext{%
  \footnotesize \textcopyright  2026 IEEE. Personal use of this material is permitted. Permission from IEEE must be obtained for all other uses, in any current or future media, including reprinting/republishing this material for advertising or promotional purposes, creating new collective works, for resale or redistribution to servers or lists, or reuse of any copyrighted component of this work in other works. Digital Object Identifier: \href{https://ieeexplore.ieee.org/abstract/document/11355735}{10.1109/TWC.2026.3652460}}

  \newcommand\copyrightnotice{%
\begin{tikzpicture}[remember picture,overlay]
\node[anchor=north,yshift=0pt] at (current page.north) {\fbox{\parbox{\dimexpr0.95\textwidth-\fboxsep-\fboxrule\relax}{\copyrighttext}}};
\end{tikzpicture}%
}

\newacronym{RIS}{RIS}{reconfigurable intelligent surface}
\newacronym{QoS}{QoS}{quality of service}
\newacronym{LC}{LC}{liquid crystal}
\newacronym{SNR}{SNR}{signal to noise ratio}
\newacronym{TDMA}{TDMA}{time-division multiple-access}
\newacronym{SDMA}{SDMA}{spatial-division multiple-access}
\newacronym{SDR}{SDR}{semidefinite relaxation}
\newacronym{BS}{BS}{base station}
\newacronym{MU}{MU}{mobile user}
\newacronym{NF}{NF}{near-field}
\newacronym{FF}{FF}{far-field}
\newacronym{Tx}{Tx}{transmitter}
\newacronym{Rx}{Rx}{receiver}
\newacronym{AWGN}{AWGN}{additive white Gaussian noise}
\newacronym{w.r.t.}{w.r.t.}{with respect to}
\newacronym{RDE}{RDE}{reaction-diffusion equation}
\newacronym{PDE}{PDE}{partial differential equation}
\newacronym{UPA}{UPA}{uniform planar array}
\newacronym{AO}{AO}{alternating optimization}
\newacronym{SOCP}{SOCP}{second-order cone programming}
\newacronym{AoD}{AoD}{angle of departure}
\newacronym{5G NR}{5G NR}{fifth generation new radio}
\newacronym{RV}{RV}{random variable}

\newacronym{PIN}{PIN}{positive-intrinsic-negative}
\newacronym{RF}{RF}{radio frequency}
\newacronym{MEMS}{MEMS}{micro-electro-mechanical system}

\newacronym{LOS}{LOS}{line of sight}
\newacronym{nLOS}{nLOS}{non-LOS}
\newacronym{MISO}{MISO}{multi-input single-output}
\newacronym{MIMO}{MIMO}{multiple-input multiple-output}
\newacronym{CSI}{CSI}{channel state information}
\newacronym{PDF}{PDF}{probability density function}
\newacronym{PCD}{PCD}{parallel coordinate descent}
\newacronym{mmWave}{mmWave}{millimeter wave}
\newacronym{AC}{AC}{alternating current}
\newacronym{THz}{THz}{terahertz}
\newacronym{GHz}{GHz}{gigahertz}
\begin{document}
\author{
\IEEEauthorblockN{Mohamadreza Delbari$^{\orcidlink{0000-0002-4768-5874}}$, 
Robin Neuder$^{\orcidlink{0000-0002-2909-7841}}$, 
Alejandro Jim\'{e}nez-S\'{a}ez$^{\orcidlink{0000-0003-0468-1352}}$, 
Arash Asadi$^{\orcidlink{0000-0001-9946-4793}}$, 
and Vahid Jamali$^{\orcidlink{0000-0003-3920-7415}}$ 
}
\IEEEauthorblockA{
\thanks{Delbari and Jamali’s work was supported in part by the Deutsche Forschungsgemeinschaft (DFG, German Research Foundation) within the Collaborative Research Center MAKI (SFB 1053, Project-ID 210487104) and in part by the LOEWE initiative (Hesse, Germany) within the emergenCITY Centre under Grant LOEWE/1/12/519/03/05.001(0016)/72. Neuder and Jim\'{e}nez-S\'{a}ez's work was supported by the Deutsche Forschungsgemeinschaft (DFG, German Research Foundation) – Project-ID 287022738 – TRR 196 MARIE within project C09. Asadi's work was in part supported by DFG HyRIS (455077022) and DFG mmCell (416765679). In addition, thanks goes to Merck Electronics KGaA, Darmstadt, Germany, for providing the liquid crystal mixture. This paper was presented in part at the 2024 IEEE International Conference on Communications Workshops (ICC, workshops) [DOI: \href{https://ieeexplore.ieee.org/document/10615422}{10.1109/ICCWorkshops59551.2024.10615422}] in \cite{delbari2024fast}. (\textit{Corresponding author: Mohamadreza Delbari.})\\
Mohamadreza Delbari and Vahid Jamali are with the Resilient Communication Systems Laboratory, Technische Universität Darmstadt, 64283 Darmstadt, Germany (e-mail: mohamadreza.delbari@tu-darmstadt.de;
 vahid.jamali@tu-darmstadt.de).
Robin Neuder and Alejandro Jim\'{e}nez-S\'{a}ez are with the Institute of Microwave Engineering and Photonics, Technische Universität Darmstadt, 64283 Darmstadt, Germany (e-mail: robin.neuder@tu-darmstadt.de;
 alejandro.jimenez\_saez@tu-darmstadt.de).
Arash Asadi is with the Embedded Systems Group, Delft University of Technology, 2628 CD Delft, Netherlands (e-mail: a.asadi@tudelft.nl).}}}
\title{Fast Reconfiguration of \MD{Liquid Crystal}-RISs: \\
Modeling and Algorithm Design}
\maketitle
\copyrightnotice
\begin{abstract}
\Gls{LC} technology is a promising hardware solution for realizing extremely large \glspl{RIS} due to its advantages in cost-effectiveness, scalability, energy efficiency, and continuous phase shift tunability. However, the slow response time of the \MD{\gls{LC}-\gls{RIS} phase shifters}, especially in comparison to the silicon-based alternatives like radio frequency switches and \gls{PIN} diodes, limits the performance. This limitation becomes particularly relevant in \gls{TDMA} applications where \gls{RIS} must sequentially serve users in different locations, as the phase-shifting response time of \MD{\gls{LC}-\gls{RIS} phase shifters} can constrain system performance. This paper addresses the slow phase-shifting limitation of \gls{LC} by developing a physics-based model for the time response of an \gls{LC} unit cell and proposing a novel phase-shift design framework to reduce the transition time. Specifically, exploiting the fact that \gls{LC}-\gls{RIS} at \gls{mmWave} bands have a large number of elements, we optimize the \gls{LC} phase shifts based on user locations, eliminating the need for full \gls{CSI} and minimizing reconfiguration overhead. Moreover, instead of focusing on a single point, the \gls{RIS} phase shifters are designed to optimize coverage over an area. This enhances communication reliability for mobile users and mitigates performance degradation due to user location estimation errors. 
The proposed \gls{RIS} phase-shift design minimizes the transition time between configurations, a critical requirement for \gls{TDMA} schemes. Our analysis reveals that the impact of \gls{RIS} reconfiguration time on system throughput becomes particularly significant when \gls{TDMA} intervals are comparable to the reconfiguration time. In such scenarios, optimizing the phase-shift design helps mitigate performance degradation while ensuring specific quality of service requirements. Moreover, the proposed algorithm has been tested through experimental evaluations, which demonstrate that it also performs effectively in practice.
\end{abstract}

 \begin{IEEEkeywords}
 Liquid crystal, reconfigurable intelligent surfaces, tuning time, and fast reconfiguration.
 \end{IEEEkeywords}

\glsresetall

\glslocalunset{RIS}

\section{Introduction}
\IEEEPARstart{R}{ECONFIGURABLE} intelligent surfaces (RISs) have recently emerged as a solution to realize programmable radio environments in the field of wireless communications \cite{qingqing2019IRS,di2019smart,yu2021smart}. \Gls{RIS}s comprise a large array of sub-wavelength elements capable of dynamically adjusting the phases of reflected electromagnetic waves. By intelligently controlling these phase shifts, \gls{RIS}s establish virtual links between \gls{BS} and \gls{MU}, when the direct \gls{LOS} \gls{BS}-\gls{MU} link is obstructed \cite{najafi2020physics,bjornson2020rayleigh}. These \gls{RIS}s require a large number of elements to compensate for the significant path loss caused by double-path reflection \cite{najafi2020physics}. \MD{Such electrically large surfaces introduce several unique challenges, including the need for scalable low-cost hardware architectures, low-complexity low-overhead \gls{CSI} acquisition and \gls{RIS} configuration schemes, and accurate channel models that account for \gls{NF} propagation regime.}

\begin{figure}[tbp]
    \centering
    \includegraphics[width=0.5\textwidth]{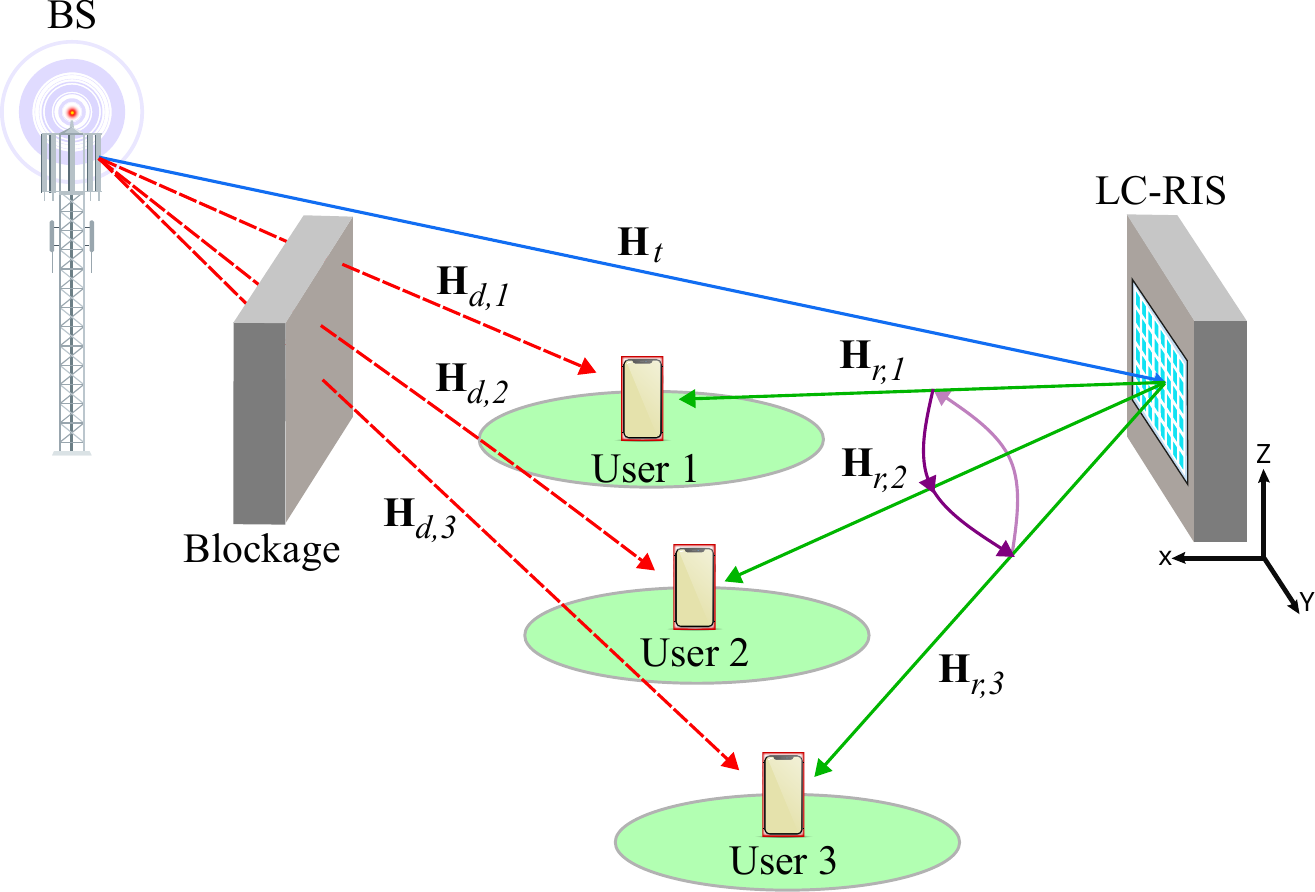}
        \caption{An \gls{RIS} assists to establish a virtual link between a transmitter and multiple receivers in a time-sequential manner, while the direct channel is blocked by an obstacle.}
\vspace{-5mm}
    \label{fig: system model}
\end{figure}
\MD{Several technologies exist for implementing RIS phase shifters, including varactor diodes \cite{tawk2012varactor,wolff2023continuous}, \gls{PIN} diodes \cite{tang2020wireless,Zeng2022,Zeng2024}, \gls{MEMS} \cite{ferrari2022reconfigurable,schmitt20213}, and \gls{LC}. Each technology presents a unique trade-off between performance, cost, and scalability. For instance, while \gls{PIN} diodes and \gls{MEMS} offer very fast response times, they typically provide only discrete phase shifts. Achieving finer phase resolution requires multiple diodes per element, which increases both cost and complexity, hindering scalability, especially at \gls{mmWave} frequencies where the number of \gls{RIS} elements is large. Varactor diodes, while offering continuous tuning, are generally limited to lower frequency ranges and are less effective for \gls{mmWave} applications. Collectively, these limitations challenge the scalability of silicon-based technologies for future \gls{mmWave} and sub-\gls{THz} systems. In contrast, \gls{LC} is notable for its cost-effectiveness, scalability, energy efficiency, and continuous phase tunability \cite{Kim2023Independtly}. Despite these advantages, the adoption of \gls{LC}-based \gls{RIS}s is constrained by their relatively slow phase response times compared to alternative technologies such as \gls{PIN} diodes, \gls{MEMS}, and varactor diodes \cite{wolff2023continuous} (See \cite[Table~I]{jimenez2023reconfigurable} for a comprehensive comparison of these technologies).}

\MD{Given their unique advantages, this paper focuses on \gls{LC}-\glspl{RIS}. We note that their relatively slow response time confines their primary application domains to quasi-static environments, where they can be leveraged to create stable and favorable propagation paths that overcome blockages or extend coverage\footnote{\MD{Due to their slow response, \gls{LC}-\glspl{RIS} are not suitable for compensating fast-fading variations at the microsecond timescale typically associated with resource block lengths. Consequently, a location-based illumination approach, as considered in this paper, provides a practical and effective strategy for exploiting the benefits of \gls{LC}-\glspl{RIS} in quasi-static environments. It is also worth noting that even for faster \gls{RIS} technologies (e.g., \gls{PIN}-diode-based designs), tracking fast fading remains highly challenging because of the prohibitive overhead associated with \gls{CSI} acquisition.}}. Examples of quasi-static radio environments include fixed wireless access and for serving low-to-moderate mobility users in indoor hotspots. Even in these targeted applications,  the reconfiguration latency remains a critical performance bottleneck, since the transition period represents an overhead during which users cannot be effectively served. In a \gls{TDMA} system, for instance, reducing this overhead translates directly into higher effective data rates. Addressing this fundamental hardware limitation is therefore essential for realizing the full potential of \gls{LC}-\glspl{RIS}, and constitutes the central focus of this paper.}


The potential of \gls{LC}-\glspl{RIS} to enhance wireless channels remains relatively underexplored within the wireless communications community. Existing studies have investigated \glspl{RIS} implemented with \gls{LC} technology from several perspectives \cite{ghannam2021reconfigurable, guirado2022mm, aboagye2022design, ndjiongue2021re, neuder2023compact, wang2004correlations, Wang2005, delbari2024temperature}. For instance, \cite{aboagye2022design,ndjiongue2021re} designed LC-based RISs specifically for visible light communication systems. From an experimental viewpoint, the behavior of \gls{LC}-RISs was analyzed in \cite{neuder2023compact}, whereas the physical mechanisms influencing the \gls{LC} response time were explored in detail in \cite{wang2004correlations, Wang2005}. Additionally, a comprehensive overview of LC-RIS properties was presented in \cite{jimenez2023reconfigurable}.

In this paper, by investigating the physical dynamics of \gls{LC} molecules, we develop a comprehensive model for their response time. Exploiting this model, we formulate an optimization problem to design RIS configurations that minimize the transition time between users. \MD{By providing a precise model and optimization framework for this fundamental hardware latency, our work provides a critical physical-layer foundation that enables more practical and efficient system integration and protocol design \cite{Lu2024}.} Our main contributions are summarized as follows:
\begin{itemize}
\item \textbf{Physics-based Modeling of LC-RIS Phase-shift Response Time:} \MD{We develop an analytical model for the response time of an LC-based phase shifter, rooted in the physical properties of its constituent \gls{LC} molecules}. Specifically, we study the \gls{PDE} that governs the spatio-temporal evolution of \gls{LC} molecules upon applying an \MD{electric field}. Using this, we develop a model for the phase shift of each \gls{LC}-\gls{RIS} unit cell as a function of $t$.
\item \textbf{Fast-transitionable Phase-shift Design:} We formulate an optimization problem to design multiple \gls{RIS} phase-shift configurations that enable fast transitions between them. \MD{As the \gls{CSI} acquisition is a challenging task for \glspl{RIS} \cite{Lu2024}, we formulate the phase-shift design problem based on user locations. The accuracy of this design increases in \gls{mmWave} since \gls{LOS} component becomes dominant \cite{Mukherjee2017}.}
Additionally, to increase the communication reliability against the location estimation error, the \gls{RIS} is designed to serve users over a designated area rather than a single point.
\item \textbf{Proposed Algorithm:} Because this problem is non-convex, finding the global optimal solution is computationally challenging. Based on the dual Lagrangian formulation and exploiting the large number of the \gls{LC}-\gls{RIS} elements, we propose an approach to decouple the optimization across users and \gls{RIS} elements, enabling parallel and iterative updates. Based on this result, we propose an algorithm that obtains a sub-optimal, but low-complexity, and scalable solution. We analyse the complexity of the proposed algorithm and show that it increases linearly with the number of \gls{LC}-\gls{RIS} elements.
\item \textbf{Performance Evaluation:} We assess the performance of the proposed \gls{RIS} design through a comprehensive set of numerical simulations. These results demonstrate that the proposed approach significantly reduces reconfiguration time (e.g., by $66\%$ for the considered setup), thereby improving the effective system throughput compared to the transition-unaware benchmark schemes.
\item \textbf{Experimental Verification:} Finally, we validate our algorithm through a real-world test in an indoor scenario \cite{neuder2023architecture}. We observe that the optimized phase shift design significantly reduces the required transition time compared to a linear phase shift design algorithm, which only focuses on maximizing the final received \gls{SNR}.
\end{itemize}
The remainder of this paper is organized as follows. In Section \ref{System, Channel, and LC Model}, we present the system, channel, and steady state of \gls{LC} models. Section \ref{Transition-Aware LC-RIS} details the proposed optimization framework based on the dynamics \gls{LC} models, followed by simulation results in Section \ref{simulation result}. Then, we test our algorithm in an experimental implementation in Section~\ref{experimental results}. Finally, Section \ref{conclusion} concludes the paper.

\textit{Notation:} Bold uppercase and lowercase letters are used to denote matrices and vectors, respectively.  $(\cdot)^\Trans$ and $(\cdot)^\Herm$ denote the transpose and Hermitian, respectively. Moreover, $\boldsymbol{0}_n$ and $\boldsymbol{1}_{n}$ denote a column vectors of size $n$ whose elements are all zeros and ones, respectively. $\bigO$ is the big-O notation.  $[\bX]_{m,n}$ and $[\bx]_{n}$ denote the element in the $m$th row and $n$th column of matrix $\bX$ and the $n$th entry of vector $\bx$, respectively. $\mathcal{CN}(\boldsymbol{\mu},\boldsymbol{\Sigma})$ denotes a complex Gaussian random vector with mean vector $\boldsymbol{\mu}$ and covariance matrix $\boldsymbol{\Sigma}$. $x!$ denotes factorial of number $x$. $\Ex\{\cdot\}$ and $\Var\{\cdot\}$ represent expectation and variance, respectively. Finally, $\Zset$, $\mathbb{R}$, and $\mathbb{C}$ represent the sets of integer, real, and complex numbers, respectively. 

\section{System, Channel, and LC-RIS Cell Model}
\label{System, Channel, and LC Model}
In this section, we first present the considered system model. Subsequently, we introduce the adopted channel model and steady-state model of \gls{LC} phase shifters.
\subsection{System Model}
\label{System model}
We consider a narrow-band downlink communication scenario comprising a \gls{BS} with $N_t$ antenna elements, an \gls{RIS} with $N$ \gls{LC}-based unit cells, and $K$ single-antenna \glspl{MU} as illustrated in Fig.~\ref{fig: system model}. The users are served in a \gls{TDMA} scheme\footnote{\MD{We adopt a simple setting consisting of a single cell and a \gls{TDMA} scheme for serving single-antenna users in order to focus primarily on the \gls{LC}-\gls{RIS} transition time between different phase-shift configurations. Nevertheless, the model proposed in Section~\ref{Transient dynamic model of LC phase shifter} for LC-RIS time response is valid for more sophisticated scenarios and the algorithm proposed in Section~\ref{Problem Formulation and the Proposed Solution} can be extended to other multiple-access schemes and multiple-antennas users \cite{Lu2025}.}}. 
The received signal at user $k$ in their allocated time slot is given by:
\begin{equation}
\label{Eq:IRSbasic}
	y_k = \big(\bh_{d,k}^\Herm + \bh_{r,k}^\Herm \bGamma_k \bH_t \big) \bx_k +n_k,\quad k=1,\dots,K,
\end{equation}
where $\bx_k\in\Cset^{N_t},\,\forall k$ is the transmit signal vector, $y_k\in\Cset$ is the received signal vector at the $k$th \gls{MU}, and  $n_k\in\Cset$ represents the \gls{AWGN} at the $k$th \gls{MU}, i.e., $n_k\sim\sCN(0,\sigma_n^2)$, where $\sigma_n^2$ is the noise power. Assuming linear beamforming, the transmit vector $\bx_k\in\Cset^{N_t}$ can be written as $\bx_k=\bq_ks$, where $\bq_k\in\Cset^{N_t}$ is the beamforming vector for user $k$ and $s\in\Cset$ is the data symbol. Assuming $\Ex\{|s|^2\}=1$, the beamformer satisfies the transmit power constraint $\sum_{k=1}^{K}\|\bq_k\|^2\leq P_t$, where $P_t$ is the total power constraint.
Moreover,  $\bh_{d,k}\in\Cset^{N_t}, \bH_t\in\Cset^{N\times N_t}$, and $\bh_{r,k}\in\Cset^{N}$ denote the BS-UE, BS-RIS, and RIS-\gls{MU} channel matrices, respectively. Furthermore, $\bGamma_k\in\Cset^{N\times N}$ is a diagonal matrix with main diagonal entries $[\bGamma_k]_n=[\bOmega_k]_n\e^{\jj[\bomega_k]_n}$ denoting the reflection coefficient applied by the $n$th \gls{RIS} unit cell comprising phase shift $[\bomega_k]_n$  and reflection amplitude $[\bOmega_k]_n$ when the \gls{RIS} serves $k$th user in its allocated time slot. We assume throughout the paper,  $[\bOmega_k]_n=1,\,\forall n$ holds for \gls{LC}-\gls{RIS} cells regardless of $k$ \cite{yang2020design} and we can control the phase shifts of each cell, $[\bomega_k]_n$.
\subsection{Channel Model}
\label{Channel model}
Given the assumption of an extremely large \gls{RIS}, the distances between the RIS and both the \gls{BS} and the \gls{MU} may fall within the \gls{NF} regime of the RIS \cite{Liu2023nearfield,delbari2024nearfield}. Therefore, an \gls{NF} channel model is adopted. Furthermore, RISs are typically installed at elevated heights, ensuring \gls{LOS} links between the RIS and both the BS and \gls{MU}s. High-frequency further emphasizes the dominance of these LOS links over \gls{nLOS} links. As a result, the channels are modeled using Rician fading with a high $K$-factor, reflecting the significant contribution of LOS components relative to nLOS components. For the sake of simplicity, we present the model of general \MD{$\bH\in\Cset^{N_\rx\times N_\tx}$} where $N_\tx$ and $N_\rx$ are the number of transmit and receive antennas, respectively.

A Rician \gls{MIMO} channel model can be written as
\begin{equation}
\label{eq: channel model}
    \bH=\sqrt{\frac{K_f}{K_f+1}}\bH^{\mathrm{LOS}}+\sqrt{\frac{1}{K_f+1}}\bH^{\mathrm{nLOS}},
\end{equation}
where $K_f$ denotes the $K$-factor
and determines the relative power of the \gls{LOS} component to the \gls{nLOS} components of the channel. \MD{In \gls{NF} regime, $\bH^\LOS$ and $\bH^\nLOS$ are functions of the locations of antennas in order to capture the underlying spherical wave propagation. This leads to \cite{Liu2023nearfield}}
\begin{IEEEeqnarray}{cc} 
	[\bH^\LOS]_{m,n} = \, c_0\e^{\jj\kk\|\bu_{\rx,m}-\bu_{\tx,n}\|}\label{Eq:LoSnear},\\
    \bH^\nLOS_s=c_s\ba_{\rx}(\bu_{s})\ba_{\tx}^\Trans(\bu_{s}),\\
 {[\ba_{\tx}(\bu_{s})]_n}   = \e^{\jj\kk\|\bu_{\tx,n}-\bu_{s}\|}
 \,\text{and}\, 
 {[\ba_{\rx}(\bu_{s})]_m}   = \e^{\jj\kk\|\bu_{\rx,m}-\bu_{s}\|}\!\!,\quad\label{Eq:nLoSnearPoint}
\end{IEEEeqnarray}
\MD{where $\bH^\LOS$ denotes the LOS \gls{NF} channel matrix and is only a function of antenna locations in \gls{Tx} and \gls{Rx}}, $c_0$ represents the channel amplitude of the LOS path, and $\bu_{\tx,n}$ and $\bu_{\rx,m}$ are the locations of the $n$th \gls{Tx} antenna and the $m$th \gls{Rx} antenna, respectively. Moreover, $\kk=2\pi/\lambda$ is the wave number with $\lambda$ being the wavelength, $\ba_{\tx}(\cdot)\in \Cset^{N_{\tx}}$ and $\ba_{\rx}(\cdot)\in \Cset^{N_{\tx}}$ denote the \gls{Tx} and \gls{Rx} \gls{NF} array response, respectively, $\bu_{s}$ is the location of the $s$th scatterer, and $c_s$ denotes the end-to-end amplitude of the $s$th non-LOS path \cite{delbari2024far}.

For the \gls{LOS} link between the \gls{BS} and the \gls{MU}s, we incorporate an additional high penetration loss due to a blockage, as suggested in \cite{Phillips2013}. Consequently, in the optimization section, we assume $\bh_{d,k} \approx \bzero$ for all $k$ to simplify the analytical analysis while we account for its impact in the simulation results in Section~\ref{simulation result}.
\subsection{Steady-state Model of LC Phase Shifter}
\label{Steady-state model of LC phase shifter}
\begin{figure}[tbp]
	\centering
    \includegraphics[width=65mm]{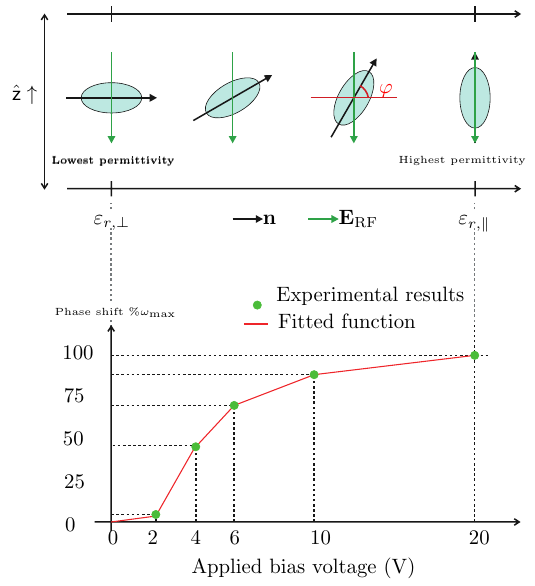}
    \caption{Phase shift vs. applied voltage for LC with 4.6 $\mu m$ LC layer thickness. The experimental data is taken from \cite{neuder2023architecture}, being fitted by a piecewise linear function.}
    \vspace{-5mm}
    \label{fig:V_phase}
\end{figure}
The anisotropy of LC molecules allows the control of the phase shift of \gls{RF} signals by adjusting the orientation of the \gls{LC} molecules.
 To achieve this effect, a thin \gls{LC} layer, typically on the order of micrometers, is placed between two electrodes. As illustrated in Fig. \ref{fig:V_phase}, in the absence of an applied voltage (a low kHz \gls{AC} voltage for biasing), i.e., $E=0$, the \gls{LC} molecules remain in a relaxed state. As a result, the radio-frequency electric field $\bE_\RF$ is perpendicular to the unit vector $\bn$ ($\|\bn\|=1$), representing the average molecular alignment in a nematic LC\footnote{\MD{ The nematic phase is a state of LCs where the rod-like molecules exhibit long-range orientational order (i.e., they tend to point in the same general direction) but lack long-range positional order. Crucially for RIS applications, this collective orientation can be manipulated by applying an external electric field ($\bE_\RF$).}}. This leads to a minimum permittivity, denoted by $\varepsilon_{\perp}=\varepsilon_{r,\perp}\varepsilon_0$, where $\varepsilon_{r,\perp}$ and $\varepsilon_0$ are minimum relative and vacuum permittivities, respectively. Conversely, under maximum applied voltage ($E=E_{\max}$), most of the \gls{LC} molecules align with the induced external electric field. This alignment causes $\bE_\RF$ to become parallel to $\bn$, resulting in achieving maximum permittivity, denoted by $\varepsilon_{\parallel}=\varepsilon_{r,\parallel}\varepsilon_0$, where $\varepsilon_{r,\parallel}$ is maximum relative permittivity. The maximum achievable phase shift, $\Delta\omega_{\max}$, is proportional to the maximum differential wave number in birefringent\footnote{\MD{ Birefringence is the property of a material where its permittivity depends on the polarization of the EM wave relative to the material's structure. For LCs, this means the permittivity experienced by a wave is different when it is polarized parallel (high permittivity) versus perpendicular (low permittivity) to the long axis of the LC molecules, which is the key mechanism for tuning the phase shift.}} materials, e.g., \gls{LC}s, denoted as $\Delta \kk_{\max}$, and can be expressed as:
\begin{equation}
\Delta\omega_{\max}=\ell_\lc\Delta \kk_{\max},
\label{eq: phase shifter function of kappa}
\end{equation}
where $\ell_\lc$ is the length of the phase shifter. Here, $\Delta \kk_{\max}$ is defined as $\kk_{\parallel}-\kk_{\perp}$, where the wave number $\kk_{g}$ is given by $2\pi f\sqrt{\varepsilon_g\mu_g},\, g\in\{\perp,\parallel\}$. Here, $f$ and $\mu_g$ represent frequency and permeability, respectively. For \gls{LC} materials, $\mu_\perp=\mu_\parallel=\mu_0$ typically holds \cite{garbovskiy2012liquid}, where $\mu_0$ is vacuum permeability. Assuming the phase shift corresponding to $\kk_{\perp}$ as the reference (i.e., $\omega=0$), the effective phase shift that each cell can apply in addition to the reference phase shift is given by:
 \begin{equation}
 \label{eq: omega and kappa function static}
    \omega=\ell_\lc(\kk-\kk_{\perp}),
\end{equation}
 where $\kk_{\perp}\leq \kk\leq \kk_{\parallel}$. The temporal evolution of $\kk$ and $\omega$ is analyzed in the next section, where their dynamic behavior is discussed in detail.

\section{Transition-Aware LC-RIS Phase-shift Design}
\label{Transition-Aware LC-RIS}
In this section, we begin by deriving the relationship between the \gls{LC} phase shifter and time using the underlying fundamental physics principles. Next, using this model, we formulate an optimization that minimizes the transition time among phase-shift configurations and develop an efficient low-complexity solution to this problem.
\subsection{Transient Dynamic Model of LC Phase Shifter}
\label{Transient dynamic model of LC phase shifter}
The operational principle of an \gls{LC} molecule is based on its electromagnetic anisotropy, wherein the electromagnetic characteristics of the \gls{LC} molecule depend on its orientation relative to the \gls{RF} electric field $\bE_\RF$ \cite{jimenez2023reconfigurable}. Due to its ellipsoidal molecular shape, an \gls{LC} molecule exhibits higher permittivity, i.e., a greater phase shift, when the electric field $\bE_\RF$ is aligned with the molecule's major axis $\bn$, compared to when it is aligned with the molecule's minor axis, as depicted in Fig.~\ref{fig:V_phase}. The time required to achieve a certain angle (vector $\bn$ \gls{w.r.t.} $\bE_\RF$) depends on the applied electrically induced force at the electrodes as well as the mechanical characteristics of the \gls{LC} molecules. In the following, we first recap the fundamental properties of the LC molecules which is the basis for derivation of the proposed \gls{LC} response time model in Section~\ref{LC Phase shifter response time}.
\subsubsection{Three-dimensional system} 
\begin{figure}
\begin{subfigure}{0.15\textwidth}
    \centering
    \includegraphics[width=1\textwidth]{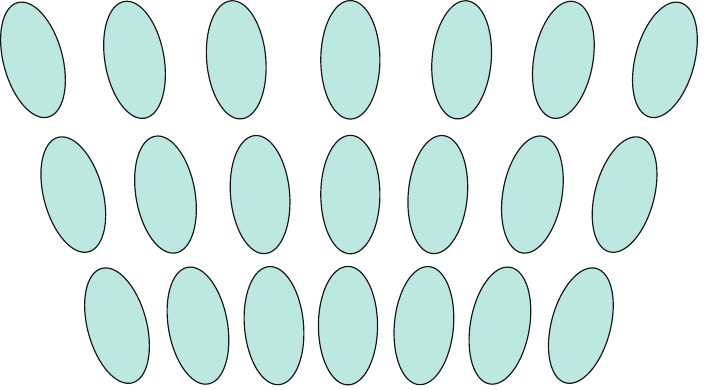}
    \caption{Spray ($K_{11}$)}
    \label{fig:Spray}
\end{subfigure}
\hfill
\begin{subfigure}{0.10\textwidth}
    \centering
    \includegraphics[width=1\textwidth]{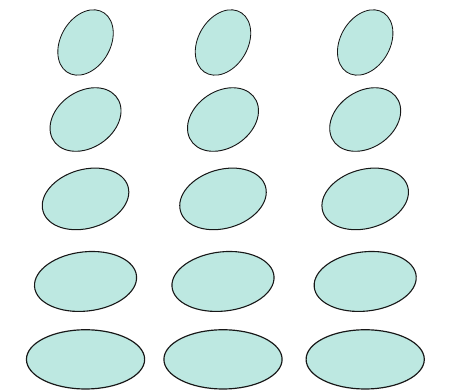}
    \caption{Twist ($K_{22}$)}
    \label{fig:Twist}
\end{subfigure}
\hfill
\begin{subfigure}{0.2\textwidth}
    \centering
    \includegraphics[width=1\textwidth]{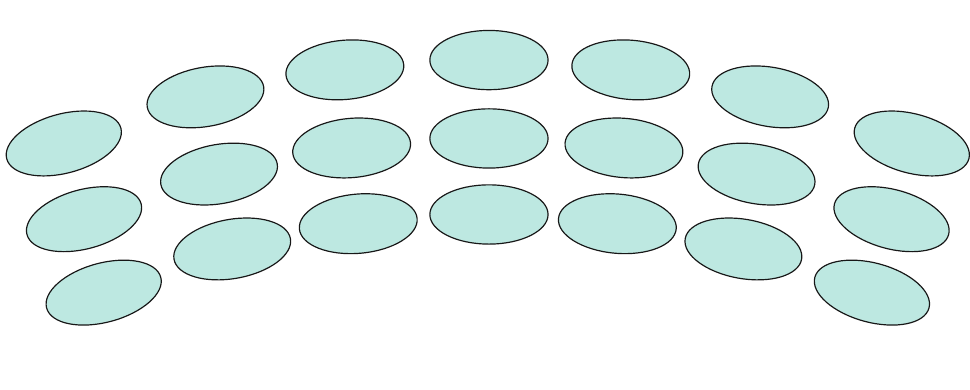}
    \caption{Bend ($K_{33}$)}
    \label{fig:Bend}
\end{subfigure}
\caption{Demonstrations of three basic deformations of \gls{LC}.}
\label{fig: spray twist bend}
\vspace{-5mm}
\end{figure}
Consider \gls{LC} molecules with a unit vector $\bn=[n_x,n_y,n_z]^\Trans$. The time evolution of this director field $\bn$ can be described by the principle of dissipative dynamics. This principle states that the rate of change of the system is proportional to the total forces (or torques) arising from the gradient of the free energy, balanced by viscous dissipation \cite{Selinger2024}:
\begin{equation}
\label{eq: elastic with time}
    \gamma_1 \frac{\partial \bn}{\partial t} = - \frac{\delta F}{\delta \bn} + \lambda_\lc \bn.
\end{equation}
Here, $\frac{\partial \bn}{\partial t}$ is the time derivative of the director field, and $- \frac{\delta F}{\delta \bn}$ is the functional derivative of the free energy \gls{w.r.t.} $\bn$, which gives the torque acting on $\bn$. $\gamma_1$ and $\lambda_\lc$ are rotational viscosity coefficient quantifying resistance to director motion, and Lagrange multiplier enforcing the unit-vector constraint $\|\bn\|=1$, respectively. Assuming no fluid flow and ignoring backflow and inertial effects \cite{Wang2005}, the total free energy $F$ of the nematic LC is the integral over the volume $V$ representing the volume of LC molecules:
\begin{equation}
    F=\int_V f_\elastic+f_\electric \dd v,
\end{equation}
where $f_\elastic$ and $f_\electric$ are the Frank elastic energy density and the electric energy density, respectively. The Frank elastic energy density captures the cost of deforming the director field $\bn$ as follows \cite{frank1958liquid},
\begin{equation}
    \label{eq: elastic}
    f_\elastic=\frac{1}{2} K_{11} \|\nabla \cdot \bn\|^2 + \frac{1}{2} K_{22} \|\bn \cdot \nabla \times \bn\|^2 + \frac{1}{2} K_{33} \|\bn \times \nabla \times \bn\|^2,
\end{equation}
where $K_{11}$, $K_{22}$, and $K_{33}$ are the elastic constants associated with spray, twist, and bend deformation, respectively, see Fig.~\ref{fig: spray twist bend}. When an electric field $\bE=[E_x,E_y,E_z]^\Trans$ is applied, it interacts with the LC molecule, aligning it in the direction of the field. The electric energy density is:
\begin{equation}
    \label{eq: energy density}
    f_\electric = -\frac{1}{2} \bE^\Trans\cdot\bvarepsilon \cdot \bE,
\end{equation}
where $\bvarepsilon$ is the dielectric permittivity tensor, and for a uniaxial LC, $\bvarepsilon$ can be expressed as:
\begin{equation}
    \label{eq: epsilon}
    \bvarepsilon = \varepsilon_\perp \bI + \underbrace{(\varepsilon_\parallel - \varepsilon_\perp)}_{\Delta\varepsilon} \bn\bn^\Trans.
\end{equation}
Here, $\bI$ denotes an identity matrix. Substituting \eqref{eq: elastic} and \eqref{eq: energy density} in \eqref{eq: elastic with time} has in general a complicated form.
However, for the \gls{LC}-\gls{RIS} consisting of LC molecules encapsulated between two large planes can be simplified to an one-dimensional \gls{PDE}. In the following, we focus on a single spatial dimension ($\vec{z}$) by assuming uniformity in the other dimensions, thereby simplifying the problem while retaining the essential characteristics.
\subsubsection{One-dimensional system} Consider an electric field applied along the $\z$-axis, with LC molecules rotating \gls{w.r.t.} this direction. Substituting the one dimension of Frank elastic energy density in \eqref{eq: elastic} and the electric energy density in \eqref{eq: energy density} into \eqref{eq: elastic with time} yields the following equation \cite{ericksen1961conservation,leslie1968some,Wang2005}:
\begin{align}
    \label{eq: elastic 1D}
        &\gamma_1\frac{\partial\varphi}{\partial t}=(K_{11}\cos^2\varphi+K_{33}\sin^2\varphi)\frac{\partial^2 \varphi}{\partial z^2}\\
    &+(K_{33}-K_{11})\sin\varphi\nonumber\times \cos\varphi(\frac{\partial\varphi}{\partial z})^2 
    +\varepsilon_0\Delta\varepsilon E^2\sin\varphi\cos\varphi,
\end{align}
where $E$ is the applied electrical field in $\z$ direction and $\varphi$ is the angle of between major axis of the LC molecule and $\x-\y$ plane, see Fig. \ref{fig:V_phase}. Even in this simplified form, \eqref{eq: elastic 1D} is not straightforward to be solved analytically, which we need to draw insights for system design. To address this issue, we adopt two assumptions, namely $K_{11}=K_{33}=\bar{K}$ \cite{Wang2005} and
\begin{equation}
    \sin\varphi\cos\varphi\approx\Phi(\varphi)\defeq\varphi(1-\frac{\varphi^2}{2}),\,0\leq\varphi\leq\frac{\pi}{2},
\end{equation}
reducing it to:
\begin{equation}
        \label{eq: elastic 1D simple}
\gamma_1\frac{\partial\varphi}{\partial t}=\bar{K}\frac{\partial^2 \varphi}{\partial z^2}+\varepsilon_0\Delta\varepsilon E^2\Phi(\varphi).
\end{equation}
The \gls{PDE} in \eqref{eq: elastic 1D simple} governs the spatio-temporal evolution of the \gls{LC} molecule alignment under an applied electric field. We will use this equation in Section~\ref{LC Phase shifter response time} as the basic to develop a model for the phase shift response time of LC-RISs, which is then used to design a fast reconfigurable phase shift in \gls{TDMA} application. But before then, in order to fully describe $\varphi$ as a function of $t$ and $z$, appropriate boundary conditions are required, which will be discussed in the following section.
\subsubsection{Boundary and initial conditions and solution}
Typically, \gls{LC}-\glspl{RIS} consist of parallel planes with an LC layer positioned between them. To solve \eqref{eq: elastic 1D simple}, we require boundary and initial conditions. Consider two planes located at $z=0$ and $z=d$, where $d$ is the \MD{\gls{LC}-\gls{RIS} phase shifter} thickness. Due to anchoring forces, the LC molecules near these planes satisfy $\varphi\approx0$. Taking into account the initial state, the angle $\varphi(z,t)$ can be described as the solution to the \gls{PDE} in \eqref{eq: elastic 1D simple}, subject to its boundary and initial conditions as follows:
\begin{subequations}
\begin{align}
        \label{eq: PDE 1D simple conditions}
    &~\text {C1:} ~~\varphi(0,t) = 0,
    \\&~\text {C2:} ~~\varphi(d,t) = 0,
    \\&~\text {C3:} ~~\varphi(z,0)=A_1\sin(\frac{\pi z}{d}),
\end{align}
\end{subequations}
where C3 has been confirmed through experimental observations, as detailed in \cite{jakeman1972electro,Wang2005} and $0<A_1<1$ is a constant. The solutions of this \gls{PDE} in two extreme cases $E=0$ and $E=E_{\max}$ are provided in Lemma~\ref{lem: PDE}.
\begin{lem}
\label{lem: PDE}
    The solutions to \eqref{eq: elastic 1D simple}, subject to the constraints C1, C2, and C3 at two extreme cases $E=0$ and $E=E_{\max}$, can be expressed as
\begin{equation}
\label{eq: varphi final}
    \varphi(z,t)=\varphi(z,\infty)+(\varphi(z,0)-\varphi(z,\infty))\e^{-\frac{t}{\tau_\mol}},
\end{equation}
with an error bound of $\bigO(e)$ with assuming $e\defeq E_{\max}^{-2}\frac{\pi^2\bar{K}}{d^2\varepsilon_0\Delta\varepsilon}\ll1$. Here, the LC molecule exponential time constant $\tau_\mol\in\{\tau_\mol^-,\tau_\mol^+\}$ depends on whether $\varphi$ is increasing or decreasing, where $\tau_\mol^-\gg\tau_\mol^+$ holds.
\end{lem}
\begin{proof}
    The proof is provided in Appendix~\ref{app: PDE proof}.
\end{proof}
The assumption $e\ll1$ is valid for typical physical parameters \cite{blinov2012electrooptic}. We have shown that the average angle of the LC molecules $\varphi$ evolves as an exponential function of time at two extreme cases. Based on this result, the next section explores in detail the time response characteristics of the LC phase shifter.
\subsubsection{LC phase shifter response time}
\label{LC Phase shifter response time}
 In Proposition~\ref{prop: linearity in time}, we establish that while $\varphi$ follows an exponential function \gls{w.r.t.} time, $\omega$ can be expressed as a summation of multiple exponential terms over time.
\begin{prop}
\label{prop: linearity in time}
    Based on the result of the Lemma~\ref{lem: PDE}, the LC phase shifter $\omega$ evolves as a summation of exponential functions over time, given by
    \begin{equation}
    \label{eq: sum integral final in prop}
      \omega(t)=\ell_\lc\sum_{p=0}^{\infty}D_p\e^{-\frac{pt}{\tau_\mol}},
    \end{equation}
where $D_p$ are constant coefficients \gls{w.r.t.} time.
\end{prop}
\begin{proof}
    The proof is provided in Appendix~\ref{app: linearity in time}.
\end{proof}
There are two important observations about the summation in \eqref{eq: sum integral final in prop}, which we use to develop a simple approximation:
\begin{itemize}
    \item The term  $\e^{-\frac{pt}{\tau_\mol}}$  decays more rapidly for larger  $p$.
    \item The coefficients $D_p$ tend to decrease as $p$ increases due to the behavior of \eqref{eq: sum integral}.
\end{itemize}
Motivated by these insights, we approximate \eqref{eq: sum integral final in prop} with a single exponential function with an exponent that is not necessarily identical to $\tau_\mol$ and can be found to improve the approximation accuracy by fitting it to measurement data. The resulting model for the \gls{LC}-\gls{RIS} time response is parameterized as follows:
\begin{equation}
\label{eq:general time response}
    [\bomega(t)]_n=[\bomega_d]_n+([\bomega_0]_n-[\bomega_d]_n)e^{\frac{-t}{[\btau_{c}]_n}},
\end{equation}
where $[\bomega_d]_n$ and $[\bomega_0]_n$ are the desired and initial phase shift applied in $n$th \gls{LC}-\gls{RIS}, respectively. In addition, $[\btau_{c}]_n\in\{\tau_c^+,\tau_c^-\}$ is the LC
director reorientation time constant, which is a function of LC parameters $\gamma$, $\bar{K}$, and $E$. It is equal to $\tau_c^+$ if $[\bomega_d]_n\geq [\bomega_0]_n$ and $\tau_c^-$ if $[\bomega_d]_n < [\bomega_0]_n$ \cite{Wang2005}. To enhance the accuracy of our proposed model, we incorporate experimentally obtained values for the exponent factor. The experimental values for $\tau_c^+$ and $\tau_c^-$ have been reported in \cite{neuder2023architecture}, see Fig. \ref{fig:time response+-}. This figure suggests that the proposed model with tuned parameters can indeed follow the experimental data.

\begin{figure}
    \centering
\includegraphics[width=0.5\textwidth]{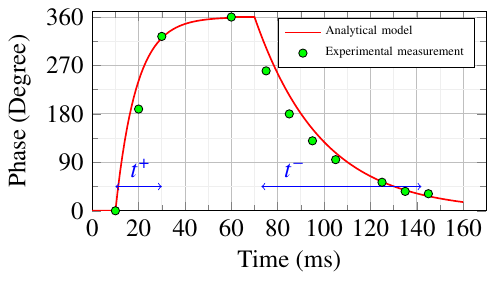}
    \caption{Experimental result of an LC phase shifter response time (green points) \cite[Fig.~4a]{neuder2023architecture}, and analytical model (red line). The time from 10\% (90\%) to reach 90\% (10\%) of the desired phase is 15~ms (72~ms) for positive (negative) phase shifts, corresponding to exponent $\tau^+_c=9$~ms ($\tau^-_c=29$~ms).}
    \label{fig:time response+-}
    \vspace{-5mm}
\end{figure}

\subsubsection{Over-and undershooting technique} 
\label{Over-and undershooting technique}
In order to accelerate the time response of \textit{each individual} \MD{\gls{LC}-\gls{RIS} phase shifter}, we can exploit so-called techniques named over- and undershooting \cite{sayyah1992anomalous}. To do that, when the phase of the \gls{LC} must increase, i.e., $[\omega_d]_n>[\omega_0]_n$ the $V_{\max}$ is applied to the corresponding cell. In contrast, when it needs to be decreased, i.e., $[\omega_d]_n<[\omega_0]_n$, zero voltage which is related to relaxed state is applied. In both cases, we change the voltage control to the $[\bv_d]_n$ when the phase reaches the desired value. The required time for $n$th cell to reach its desired phase shift by exploiting this technique can be reduced as follows
\begin{equation}
    \label{eq: time positive}    [\bt_r]_n=\tau^+_c\ln{\left(\frac{\omega_{\text{max}}-[\bomega_0]_n}{\omega_{\text{max}}-[\bomega_d]_n}\right)}=\tau^+_c\ln{\left(1+\frac{[\bomega_d]_n-[\bomega_0]_n}{\omega_{\text{max}}-[\bomega_d]_n}\right)},
\end{equation}
where $[\bt_r]_n$ determines the required time for $n$ cell to reach its desired phase. Similarly, to decrease the phase shift, i.e., $[\omega_d]_n<[\omega_0]_n$, we first set the voltage to zero using the undershooting technique so the required time is obtained as
\begin{equation}
    \label{eq: time negative}    
    [\bt_r]_n=\tau^-_c\ln{\left(\frac{[\bomega_0]_n-\omega_{\min}}{[\bomega_d]_n-\omega_{\min}}\right)}=\tau^-_c\ln{\left(1+\frac{[\bomega_0]_n-[\bomega_d]_n}{[\bomega_d]_n-\omega_{\min}}\right)},
\end{equation}
where we usually assume $\omega_{\min}=0$.
\subsection{Problem Formulation and the Proposed Solution}
\label{Problem Formulation and the Proposed Solution}
In the following, we employ the derived model in Section~\ref{Over-and undershooting technique} for the LC-RIS time response in order to develop phase shifters that enable fast reconfigurations. We consider a \gls{TDMA} setup where, in each time slot, the \gls{RIS} is configured to serve an area that includes the user, rather than focusing on a single point as in \cite{delbari2024fast}. \MD{Serving an area instead of a specific point improves communication reliability because the \gls{RIS} can maintain the required \gls{QoS} despite user movement within the area or in the presence of location estimation errors \cite[Fig.~3]{Jamali2022lowtozero}}. Achieving this objective necessitates deriving multiple phase-shift configurations that not only maximize the minimum SNR across the target area but also ensure that the time required to switch between configurations remains minimal. We define \gls{SNR} of user $k$ as follows
\begin{IEEEeqnarray}{ll}
\label{eq: SNR}
    \SNR_k(\bu_k)=\frac{|\bh_k^\Herm(\bu_k)\bq_k|^2}{\sigma^2_n}, 
\end{IEEEeqnarray}
where 
$\bh_k^\Herm(\bu_k)=\bh_{r,k}^\Herm(\bu_k) \bGamma_k \bH_t$ represents the end-to-end channel at the $\bu_k\in\Uset_k$ location, $\Uset_k$ denotes the set of possible locations of user $k$ where $k=1, \cdots, K$, and $\bq_k$ denotes the beamforming vector in \gls{BS} for user $k$. Since user location data may be imprecise, we define the \gls{QoS} metric as:
\begin{IEEEeqnarray}{ll}
\label{eq: min SNR}
\SNR_k^{\min}\triangleq\underset{\bu_k\in\Uset_k}{\min}\frac{|\bh_k^\Herm(\bu_k)\bq_k|^2}{\sigma^2_n}. 
\end{IEEEeqnarray}
The transition time for the $n$th cell from $[\omega_0]_n$ to $[\omega_d]_n$ is determined by the differential between these two phase shifts. Our objective is to achieve phase shifts that satisfy the required \gls{SNR}, $\gamma_k^\thr$, for any user location within the target area, while minimizing the time needed for reconfiguration. Assuming a given order for the serving the users\footnote{In a multi-user scenario where $K>2$, the order in which users are served affects the overall reconfiguration process, i.e., there are $(K-1)!$ possible permutations. In section~\ref{simulation result}, we investigate the impact of user order through simulations.}, we formulate the following optimization problem
\begin{subequations}
\label{eq:optimization 1}
\begin{align}
    \text {P1:}\quad&~\underset{\bomega_k,\bq_k,\,\,\forall k}{\min}~\sum_{k=1}^K\underset{n}{\max}~[\bt_k]_n
    \label{eq:optimization 1 a}\\&~\text {s.t.} ~~\frac{|\bh_k^\Herm(\bu_k)\bq_k|^2}{\sigma^2_n}\geq \gamma_k^\thr,\,\,, \forall \bu_k\in\Uset_k,\,\forall k,
    \label{eq:optimization 1 b}\\&\quad\hphantom {\text {s.t.} } 0\leq [\bomega_k]_n < \omega_\tmax, \forall n, k,
    \label{eq:optimization 1 c}\\&\quad\hphantom {\text {s.t.} } \sum_{k=1}^K\|\bq_k\|_2^2\leq P_t, \forall k,
    \label{eq:optimization 1 d}
\end{align}
\end{subequations}
where each $[\bt_k]_n, \forall k$ is defined in \eqref{eq: time positive} and \eqref{eq: time negative} \gls{w.r.t.} current user $k$ and previous user $k-1$. $\bomega_k\in\Cset^N$ and $\bq_k\in\Cset^{N_t}$ are phase shift and beamforming variables, respectively. Here, \eqref{eq:optimization 1 b} is the minimum $\SNR_k$ constraint for $k$th user to be served by \gls{RIS}. (\ref{eq:optimization 1 c}) is the realizable phase shift range where $\omega_\tmax$ shows the maximum achievable phase shift by each \gls{LC} element in the \gls{RIS}. (\ref{eq:optimization 1 d}) forces the transmit power in the \gls{BS}. The problem P1 is a non-convex problem due to the non-convexity of the objective function and constraint (\ref{eq:optimization 1 b}). Two variable vectors $\bomega_k$ and $\bq_k$ are also coupled in (\ref{eq:optimization 1 b}). Hence, obtaining the global solution for this problem is challenging in term of complexity. Due to these challenges, we adopt \gls{AO}, where the problem is separated into two sub-problems where in each one, we fix one of the variable vectors and optimize the problem \gls{w.r.t.} the other variable vector.\\
\begin{figure}
    \centering
    \includegraphics[width=0.4\linewidth]{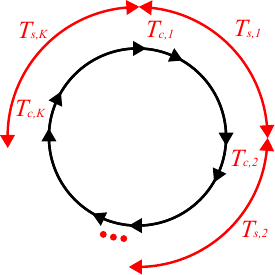}
    \caption{Configuring time ($T_c$) and switching time ($T_s$) are shown for each user when they are served by an \gls{RIS} sequentially.}
    \label{fig:sequentially}
    \vspace{-5mm}
\end{figure}

\textbf{Beamformer design:} 
\MD{For a fixed RIS phase shift $\bomega_k,\,\forall k$, the problem P1 reduces to a feasibility search for the beamformer $\bq_k$. One approach to obtain a feasible solution is to choose to maximize SNR in $\bq_k$ for the given $\bomega_k$. By doing so, we create the largest possible SNR margin above the required threshold $\gamma_k^\thr$. This SNR margin is beneficial in the next step of the \gls{AO}, where $\bq_k$ is fixed and $\bomega_k$ is optimized. A larger margin means that the SNR constraint remains satisfied for a wider range of $\bomega_k$ values, thus hopefully expanding the feasible set for the RIS phase-shift optimization (there is no guarantee). This provides the algorithm with more flexibility to find a $\bomega_k$ that minimizes the primary objective, which is the transition time. Hence, $\forall k = \{1, ..., K\},$ the following problem is solved:}
\begin{subequations}
\label{eq:optimization beamformer}
\begin{align}
    \text{P2:}\quad&~\underset{\bq_k}{\max}~\SNR_k^{\min}
    \\&\quad\hphantom {\text {s.t.} } \sum_{k=1}^K\|\bq_k\|_2^2\leq P_t,
\end{align}
\end{subequations}
where $\SNR_k^{\min}$ was defined in \eqref{eq: min SNR}. Under the assumptions of blockage, i.e., $\bh_{d,k}\approx\bzero,\,\forall k$, and dominant \gls{LOS} BS-RIS link, i.e., $K_f\to+\infty$, we can rewrite channel as follows:
\begin{equation}
    \label{eq: channel reformulation}
    \bh_{k}(\bu_k)=\bh_{r,k}^\Herm\bGamma_k\ba_\RIS(\bu_\BS)\ba_\BS^\Herm(\bu_\RIS),
\end{equation}
where $\ba_\BS(\cdot)$ and $\ba_\RIS(\cdot)$ are steering vectors in \gls{BS} toward the \gls{RIS} and \gls{RIS} toward the \gls{BS}, respectively. They were defined in \eqref{Eq:nLoSnearPoint} by substituting $\bu_s$ with $\bu_\RIS$ and $\bu_\BS$ which are the locations of the \gls{RIS} and \gls{BS}.
This is now a standard SNR maximization problem in \gls{MISO} systems, which has the following matched-filter precoder solution \cite{tse2005fundamentals}:
\begin{equation}
\label{eq: beamformer design}
\bq_k\overset{(a)}{\approx}\frac{\sqrt{P_t}}{\|\ba_\BS(\bu_\RIS)\|}\ba_\BS(\bu_\RIS).
\end{equation}
This problem is feasible when $\underset{\bu_k\in\Uset_k}{\min}P_t\frac{\|\bh_k(\bu_k)\|^2}{\sigma_n^2}\geq\gamma_k^\thr,\,\forall k$.
Because the chosen beamformer does not depend on the RIS phase shifter, we denote this specific beamforming vector as $\bq_\LOS$ and from now, we focus on optimizing the \gls{RIS} phase shifts.\\

\textbf{RIS design:} For the fixed BS beamformer given in \eqref{eq: beamformer design} the phase configuration problem is given by:
\begin{subequations}
\label{eq:optimization 3}
\begin{align}
    \text {P3:}\quad&~\underset{\bomega_k,\,\,\forall k}{\min}~\sum_{k=1}^K\underset{n}{\max}~[\bt_k]_n
    \\&~\text {s.t.} ~~\SNR_k^{\min}\geq \gamma_k^\thr,\,\, \forall k
    \\&\quad\hphantom {\text {s.t.} } 0\leq [\bomega_k]_n < \omega_\tmax, \forall n, k.
\end{align}
\end{subequations}
First, we redefine the $\SNR_k$ in \eqref{eq: SNR} \gls{w.r.t.} \gls{RIS} phase shifts:
\begin{equation}
    \label{eq: SNR new}
    \SNR_k=\underbrace{m_k+\bv_k^\Herm\bbm_k+\bbm_k^\Herm\bv_k}_{\approx0\text{ because the \gls{BS}-\gls{MU} link is blocked}}+\bv_k^\Herm\bbM_k\bv_k,
\end{equation}
where $\bv_k=[\e^{j[\bomega_k]_1}, ..., \e^{j[\bomega_k]_N}]^\Trans$ is phase shift vector in \gls{RIS} for user $k$, $\bbm_k^\Trans=\frac{(\bh_{d,k}^\Herm\bq_k)\bq_k^\Herm\bH_t^\Herm\diag(\bh_{r,k}^\Herm)}{\sigma_n^2}$, and $\bbM_k=\frac{\diag(\bh_{r,k})\bH_t\bq_k\bq_k^\Herm\bH_t^\Herm\diag(\bh_{r,k}^\Herm)}{\sigma_n^2}$. For each point within the area allocated to user $k$, the $\SNR$ expressed as $\SNR_k(\bu_k)=\bv_k^\Herm\bbM_k(\bu_k)\bv_k, \forall \bu_k\in\Uset_k$, where each element corresponds to a specific $\bu_k\in\Uset_k$. Next, we introduce an auxiliary variable $t_k^{\max}$ for each user, which satisfies the constraint $t_k^{\max}\geq[\bt_k]_n, \forall n\in\{1, \cdots, N\}$. Using this formulation, the problem P3 can be reformulated as:

\begin{subequations}
\label{eq:optimization 4}
\begin{align}
    \text {P4:}\quad&~\underset{t_k^{\max},\bomega_k,\,\,\forall k}{\min}~\sum_{k=1}^K t_k^{\max}
    \label{eq:optimization 4 a}\\&~\text {s.t.} ~~\bv_k^\Herm\bbM_k(\bu_k)\bv_k\geq \gamma_k^\thr,\,\, \forall \bu_k\in\Uset_k,\forall k,
    \label{eq:optimization 4 b}\\&\quad\hphantom {\text {s.t.} } [\bt_k]_n\leq t_k^{\max},\,\, \forall n,k,
    \label{eq:optimization 4 c}\\&\quad\hphantom {\text {s.t.} } 0\leq [\bomega_k]_n < \omega_\tmax, \forall n, k.
    \label{eq:optimization 4 d}
\end{align}
\end{subequations}
In P4, constraints (\ref{eq:optimization 4 c}) and (\ref{eq:optimization 4 d}) are convex and linear \gls{w.r.t.} variable $[\bomega_k]_n,\,\forall k$, respectively. However, (\ref{eq:optimization 4 b}) is non-convex \gls{w.r.t.} variable $[\bomega_k]_n$. In the following, we develop a sub-optimal solution to problem P4 based on the Lagrangian formulation and exploiting the fact that $N$ is large. 

Let us define $\bs_k\!\triangleq\![\bv_k^\Herm\bbM_k(\bu_k^{(1)})\bv_k, \cdots, \bv_k^\Herm\bbM_k(\bu_k^{(|\Uset_k|)})\bv_k]^\Trans\!\!\!,$ where $\bu_k^{(j)}$ denotes $j$th location point of the possible locations for user $k$ in $\Uset_k$. We define the Lagrange function $L(\bW,\bLambda,\bXi)$ as \cite{cvx}
\begin{equation}
\label{eq: lagrangian}
    \sum_{k=1}^K \Big(t_k^{\max}+\blambda_k^\Trans(\gamma_k^\thr\bone_{|\Uset_k|}-\bs_k)+\bxi_k^\Trans(\bt_k-t_k^{\max}\bone_{N})\Big),
\end{equation}
where $\blambda_k\in\Rset_+^{|\Uset_k|}$ and $\bxi_k\in\Rset_+^N,\,\forall k$ denote the associated Lagrangian multipliers. Let us assume the cardinality of the allocated area for all users is equal, i.e., $|\Uset_1|= \cdots, =|\Uset_K|=|\Uset|$. Moreover, $\bLambda=[\blambda_1, \cdots, \blambda_K]\in\Rset_+^{|\Uset|\times K}$ and $\bXi=[\bxi_1, \cdots, \bxi_K]\in\Rset_+^{N\times K}$ are matrices that collect the Lagrangian multipliers. $\bW=[\bomega_1,\cdots\bomega_K]\in\Wset^{N\times K}$ where $\Wset=\{\omega\in\Rset|0\leq\omega\leq\omega_{\max}\}$. Based on these definitions, the dual Lagrangian problem can be written
\begin{subequations}
\label{eq:optimization 5}
\begin{align}
    \text {P5:}\quad&~\underset{\bLambda,\bXi}{\max}~\underset{\bW,t_k^{\max},\,\forall k}{\inf}~L(\bW,\bLambda,\bXi)
    \\&~\text {s.t.} ~~\blambda_k\geq0,\,\bxi_k\geq 0, \forall k,
    \\&\quad\hphantom {\text {s.t.} } \bW\in\Wset^{N\times K}.
\end{align}
\end{subequations}
This problem is still non-convex due to the non-convexity of each element in $\bs_k,\,\forall k$ vector in the cost function. As can be seen from \eqref{eq: lagrangian}, the phase-shift configuration of all users is coupled together in $\bt_k,\,\forall k$. To be more exact, the phase shift set of each user affects the required configuring time for the next and previous phase-shift configuration. To cope with this issue, we continue the analysis of this problem based on a few assumptions, as will be detailed below. We will show in Section~\ref{simulation result} that despite these simplifying assumptions, the proposed phase-shift design is able to significantly reduce the LC-RIS reconfiguration time.
\subsubsection{User decoupling} Our first assumption is about the decoupling user phase shifts. With this assumption, we solve P5 iteratively where at each iteration, we focus on one of the $\bomega_k,\,k=\{1, \cdots, K\}$ where all other $\bomega_{k'},\,k'\neq k$ are fixed. Let us define for $k=1,..., K$:
\begin{IEEEeqnarray}{ll}
\Delta\bomega_k\triangleq
\begin{cases}
    \bomega_k-\bomega_{k-1},\quad &\mathrm{if}\,\, k\neq 1\\
    \bomega_{1}-\bomega_K,\,\,&\mathrm{if}\,\, k=1.
\end{cases}
\end{IEEEeqnarray}
Then, we can transform the problem P5 into a decoupled problem on each user $k$. Let us define $L_k$ for $k$th user as
\begin{align}
    \label{eq: lagrangian user k} L_k\defeq&t_k^{\max}+t_{k+1}^{\max}+\blambda_k^\Trans(\gamma_k^\thr\bone_{|\Uset|}-\bs_k)\nonumber\\
    +&\bxi_k^\Trans(\bt_k-t_k^{\max}\bone_{N})+\bxi_{k+1}^\Trans(\bt_{k+1}-t_{k+1}^{\max}\bone_{N}),
    \end{align}
where $\bt_k$ and $\bt_{k+1}$ are functions of $\bomega_k$, $\bomega_{k-1}$ and $\bomega_{k+1}$. Here, we assumed  $\bomega_{k-1}$ and $\bomega_{k+1}$ are fixed and without change at each iteration. From now on, we optimize $L_k(\bomega_k,\blambda_k,\bxi_k,\bxi_{k+1})$ sequentially for users $k=1, \cdots, K$. Due to the non-convexity of the $\bs_k$, we still cannot find a tractable solution for this problem. Hence, we apply the second assumption in the following.
\subsubsection{Decoupling on each RIS element} 
The phase shifts of an LC-RIS are inherently coupled within the vector \(\bs_k\), as represented in \eqref{eq: lagrangian user k}. This coupling significantly increases computational complexity, particularly as the number of RIS elements grows. To mitigate this issue, we employ the \gls{PCD} method \cite{wright2015coordinate}. A key challenge in using this approach here is that \(\bs_k\) represents a vector over multiple spatial points, making direct optimization intractable. To address this, we adopt a strategy inspired by \cite{Jamali2022lowtozero}, where each RIS element is associated with a single representative point in the targeted area. The method consists of the following steps:\\
\textbf{Step 1: Mapping RIS elements to pre-defined points:} We adopt a strategy similar to those proposed in \cite{Jamali2022lowtozero,delbari2024far}. In particular, we assign each RIS element to a single point within the set \(\Uset_k\). This mapping simplifies the problem by ensuring a one-to-one correspondence between RIS elements and spatial points. With this mapping, we approximate the $L_k$ as
    \begin{align}
    \label{eq: lagrangian user k over n} L_k\approx&t_k^{\max}+t_{k+1}^{\max}+\sum_{n=1}^N\left([\blambda_k]_n(\gamma_k^\thr-\SNR(\bu_k^{(n)}))\right)\nonumber\\
    +&\bxi_k^\Trans(\bt_k-t_k^{\max}\bone_{N})+\bxi_{k+1}^\Trans(\bt_{k+1}-t_{k+1}^{\max}\bone_{N}),
    \end{align}
where $\bu_k^{(n)}$ represents the point in $\Uset_k$ associated with the $n$-th RIS element. Unlike \eqref{eq: lagrangian user k}, where optimization is performed over the entire set for each element, each LC-RIS element is optimized for a specific assigned point in \eqref{eq: lagrangian user k over n}. This phase shift mapping is effective when the number of RIS elements ($N$) is sufficiently large relative to the targeted coverage area. According to the $\bu_k^{(n)}$ point, there is a best phase shift for $n$th element to maximize the SNR in the targeted point. We denote that optimized phase shift for the $n$-th element and $k$th user as $\bphi_k(\bu_k^{(n)})$.\\
\textbf{Step 2: PCD optimization method:} After assigning each RIS element to a specific point, we optimize the phase shifts using the \gls{PCD} method. Given that the original function $\SNR(\bu_k^{(n)})$ depends on all elements of $[\bomega_k]_1,\,\cdots,\,[\bomega_k]_N$, we define a coordinate-wise function for each element in the following Lemma to reduce more the complexity of the optimization.
\begin{lem}
\label{theorem: fourier}
For sufficiently large $N$, blocked BS-\gls{MU} link, dominant \gls{LOS} link between both BS-RIS and RIS-\gls{MU}, and fixing all variables $[\bomega_k]_n,\,\forall n$ except one of them ($[\bomega_k]_p$), the $L_k$ can be written as
\begin{equation}
\label{eq: summation L_n}
    L_k= C_k+[L_k]_p,
\end{equation}
where $C_k$ is a constant for $k$th user and
\begin{align}
\label{eq: L_n}
    [L_k]_p=&[\bxi_{k+1}]_p[\bt_{k+1}]_p+[\bxi_k]_p[\bt_k]_p\nonumber\\
    &-[\blambda_k]_p\cos([\bomega_k]_p-\bphi_k(\bu_k^{(p)})).
\end{align}
\end{lem}
\begin{proof}
    The proof is provided in Appendix~\ref{app: theorem fourier}.
\end{proof}
 Using Lemma~\ref{theorem: fourier}, we can derive the optimal phase shifts for each element independently, reducing the problem's dimensionality. To further reduce computational complexity, we optimize all elements in parallel using the result of Lemma~\ref{theorem: fourier}. Now, we have $[L_k]_n$ in \eqref{eq: L_n}, $\forall n$ where 
\begin{equation}
    [\bt_k]_n=\begin{cases}
   \tau^-_c\ln{(\frac{[\bomega_{k-1}]_n-\omega_{\min}}{[\bomega_k]_n-\omega_{\min}})},\,\,&\mathrm{if}\,\, [\bomega_{k-1}]_n>[\bomega_k]_n,\\
    \tau^+_c\ln{(\frac{\omega_{\text{max}}-[\bomega_{k-1}]_n}{\omega_{\text{max}}-[\bomega_k]_n})},\,\,&\mathrm{if}\,\, [\bomega_{k-1}]_n<[\bomega_k]_n,
\end{cases}
\end{equation}
 and 
 \begin{equation}
     [\bt_{k+1}]_n=\begin{cases}
   \tau^-_c\ln{(\frac{[\bomega_{k}]_n-\omega_{\min}}{[\bomega_{k+1}]_n-\omega_{\min}})},\,\,&\mathrm{if}\,\, [\bomega_{k}]_n>[\bomega_{k+1}]_n,\\
    \tau^+_c\ln{(\frac{\omega_{\text{max}}-[\bomega_{k}]_n}{\omega_{\text{max}}-[\bomega_{k+1}]_n})},\,\,&\mathrm{if}\,\, [\bomega_{k}]_n<[\bomega_{k+1}]_n.
\end{cases}
 \end{equation}
We can define $[\bt_k]_n$ in a more concise way as follows:
\begin{align}
    [\bt_k]_n&=\frac{\sign([\bomega_k]_n-[\bomega_{k-1}]_n)+1}{2}\tau^+_c\ln{(\frac{\omega_{\text{max}}-[\bomega_{k-1}]_n}{\omega_{\text{max}}-[\bomega_k]_n})}\nonumber\\
    &-\frac{\sign([\bomega_k]_n-[\bomega_{k-1}]_n)-1}{2}\tau^-_c\ln{(\frac{[\bomega_{k-1}]_n-\omega_{\min}}{[\bomega_k]_n-\omega_{\min}})}.
\end{align}

The algorithm steps and complexity analysis are described in the following part.

\textbf{Algorithm and complexity analysis:} 
The proposed algorithm to obtain a suboptimal solution to P1 is summarized in Algorithm~\ref{alg:cap}. We start with an initial number for $t_k^{\max}$ and $\blambda_k,\,\forall k$. At each round ($i=1:I_{\max}$), we aim in decreasing $t_k^{\max}$, \MD{where $I_{\max}$ is the maximum number of iterations and is determined in practice according to the available budget for computation time and complexity}. In each iteration ($k=1:K$), the focus is on one user, adjusting its phase shifters within the range of $(0,\omega_{\max})$ to minimize the transition time from $\bomega_{k-1}$ to $\bomega_{k}$ and also from $\bomega_{k}$ to $\bomega_{k+1}$. \MD{With the help of Lemma~\ref{theorem: fourier}, we can optimize the phase shift of each element in parallel. To do so, first, if $[\bt_{v}]_n<t_v^{\max}$ where $v=\{k,k+1\}$, we set the value of $[\bxi_g]_n=0$. Otherwise, we set $[\bxi_g]_n=1$ leading to a higher penalty for that \gls{LC}-\gls{RIS} phase shifter element. Second, to find the minimizer of each $[L_k]_n$ to minimize \eqref{eq: L_n} per element, we perform a one-dimensional line search. Specifically, we employ a grid search over the feasible phase-shift range $(0,\omega_{\max})$. This range is discretized into $L$ uniformly spaced points, and the cost function is evaluated at each point. The phase value that yields the minimum cost is then chosen as the updated solution. We name this minimizer as $[\bomega_k^g]_n$.}
Once a new phase shift is determined for the $k$-th user, the \gls{SNR} constraint is verified. If the constraint remains satisfied, $\blambda_k$ decreases to explore additional potential solutions, otherwise, if the constraint is violated, $\blambda_k$ is increased to prioritize meeting the SNR constraint. 


\MD{The dominant computational complexity of the algorithm arises from determining $[\bomega_k^g]_n$ in each iteration. Since our grid search requires evaluating the cost function at $L$ points for each of the $N$ \gls{RIS} elements, the complexity of one full update for user $k$ is $\bigO(NL)$. Therefore, the overall dominant complexity of the algorithm is $\bigO(I_{\max}NKL)$, which is linear with the number of RIS elements $N$ and the search resolution~$L$.}
\begin{algorithm}[t]
\caption{Proposed Algorithm for Problem (P4)}\label{alg:cap}
\begin{algorithmic}[1]
\STATE \textbf{Initialize:} Derive $[\bphi_k(\bu_k^{(n)})]_n$, $\forall n,\,k$ from \cite[Eqs. (5), (6) ]{Alexandropoulos2022}. Set $\bW^{(0)}=[\bphi_1, \cdots, \bphi_K]$, $\bq_k=\bq_\LOS$, $\blambda^{(0)}_k>0$ and $t_k^{\max}$, $\forall k=1, \cdots, K$. Set $0<\alpha<1$, $I_{\tmax}$, $\Delta t_k$.
\FOR{$i=1:I_{\tmax}$}
\STATE Set $[\bxi_k]_n=1$ if $[\bt_k]_n\geq t_k^{\max}$, and 0 otherwise, $\forall k,n$.
    \FOR{$k=1:K$}
    \STATE Find $[\bomega_k]^{(i)}_n=[\bomega_k^g]_n\in(0,\omega_{\max})$ as a minimizer of $L_n,\,\forall n$ in \eqref{eq: L_n} with line search.
    \STATE Calculate ${\SNR_k^{\min}}^{(i)}=\underset{\bu_k\in\Uset_k}{\min}\SNR_k^{(i)}(\bu_k)$.
    \IF{${\SNR_k^{\min}}^{(i)}<\gamma_k^\thr$}
    \STATE Update $\bomega^{(i)}_k=\bomega^{(i-1)}_k$, and $\blambda_k^{(i)}=\frac{\blambda_k^{(i-1)}}{\alpha}$.
    \ELSE
    \STATE Update $\bomega^{(i)}_k=\bomega^{g}_k$, and $\blambda_k^{(i)}=\alpha\blambda_k^{(i-1)}$.
    \ENDIF
    \ENDFOR
    \STATE Update $t_k^{\max}=t_k^{\max}-\Delta t_k$.
\ENDFOR
\end{algorithmic}
\end{algorithm}

\section{Performance Evaluation}
\label{simulation result}
\subsection{Simulation Setup}
We employ the simulation configuration for coverage extension illustrated in Fig. \ref{fig: system model}. \MD{We assume there are three users located in different and random locations with uniform distribution inside the areas $[10\pm1,2\pm1,-5]$~m, $[10\pm1,-5\pm1,-5]$~m, and $[10\pm1,5\pm1,-5]$~m, respectively.} The BS comprises a $4\times4=16$ \gls{UPA} in $\x-\z$ axes positioned at $[40,20,5]$~m.
The \gls{RIS} comprises a \gls{UPA} located at $[0,0,0]$~m, consisting of $N_\y\times N_\z=150\times5$ elements along the $\y$- and $\z$-axes, respectively \MDD{(except in Fig. \ref{fig:TDMA_time}, where the number of \gls{RIS} elements is specified in each subfigure)}.
The element spacing for \gls{UPA}s at both the BS and \gls{RIS} corresponds to half of the wavelength. The \gls{MU}s have a single antenna. The noise variance is computed as $\sigma_n^2=WN_0N_{\rm f}$ with $N_0=-174$~dBm/Hz, $W=20$~MHz, and $N_{\rm f}=6$~dB. We assume $28$~GHz carrier frequency, the path-loss $\beta=-61$~dB at $d_0=1$~m, and $\gamma_\thr=10$~dB. \MD{Given these parameters, the users' locations fall into \gls{NF} regime of the \gls{LC}-\gls{RIS} (i.e., the distances between the RIS and the center of the coverage areas are $11.4$~m, $12.2$~m, and $12.2$~m for users 1-3, respectively, which are below the Fraunhofer distance $d_\text{FF} = \frac{2D^2}{\lambda}=120~$m,
where $D$ is the largest physical dimension of the RIS \cite{Liu2023nearfield,delbari2024far}).} Moreover, we adopt path-loss exponent $\eta = (2,2,2)$ and  $K_f=(-100,10,10)$~dB for the BS-UE, BS-RIS, and RIS-\gls{MU} channels, respectively, and additional blockage $h_\text{blk}=-40$~dB for the direct link BS-UE. The main focus of this analysis is the \gls{SNR} comparison.
\MDD{As a benchmark, we consider the analytical \gls{RIS} design based on anomalous reflection in \cite{najafi2020physics} which is unaware of the transition behavior of the \gls{LC}-RIS. This applies to all figures except Fig.~\ref{fig:TDMA_time}, where we also include \gls{SDR} \cite{qingqing2019IRS,qingqing2021} as an additional benchmark. While \gls{SDR} achieves a better steady-state \gls{SNR} than analytical methods \cite{najafi2020physics}, its complexity scales with $\bigO(N^{3.5})$, making it unscalable for large arrays. Therefore, \gls{SDR} is used only in Fig.~\ref{fig:TDMA_time}, and for the remainder of the figures, 'Benchmark' refers to the analytical method \cite{najafi2020physics}.}
We adopted $\tau_c^+=9$~ms and $\tau_c^-=29$~ms based on the data provided in \cite{neuder2023architecture}.
The parameters used in this simulation are $\alpha=0.95$, $I_\tmax=35$, $P_t=38$~dBm, $t_k^{\max}=100$~ms, and $\Delta t_k=\frac{t_k^{\max}}{I_{\max}}$~ms $\forall k$.

\begin{remk}
The MATLAB codes used to generate the simulation results in this section are publicly available online at \href{https://github.com/MohamadrezaDelbari/LC-RIS-TDMA-Journal}{\textcolor{blue}{https://github.com/MohamadrezaDelbari/LC-RIS-TDMA-Journal}}.
\end{remk}

\subsection{Convergence Behavior}
\begin{figure}
    \centering
    \includegraphics[width=0.5\textwidth]{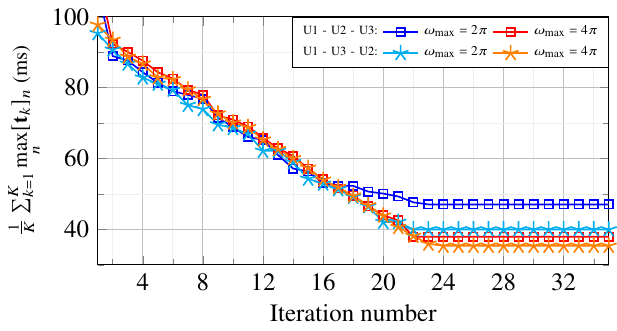}
\caption{$\frac{1}{K}\sum_{k=1}^K \underset{n}{\max}[\mathbf{t}_k]_n$ vs. number of the iteration.}
    \label{fig:deltaomega}
     \vspace{-5mm}
\end{figure}
\MD{We first discuss the convergence properties of Algorithm~\ref{alg:cap}. This algorithm is designed such that the auxiliary variables $t_k^{\max},\,\forall k$ are decreasing in each outer loop. While this does not strictly guarantee a monotonic decrease of the objective function, $\underset{n}{\max}~[\bt_k]_n,\,\forall k$, in every inner step, our simulation demonstrates that the algorithm consistently converges to a stable solution, typically within 20-30 iterations for the parameters adopted in this work.

This convergence is illustrated in Fig.~\ref{fig:deltaomega}, where} we show the averaged cost function of P1 over $K$ versus the number of iterations. As shown, these transition-time values decrease progressively with each iteration. This behavior occurs because the initial configuration ensures that $\SNR_k(\bu_k),\,\forall \bu_k\in\Uset_k$ is greater than the threshold. Furthermore, each iteration is designed to guarantee that $\SNR_k(\bu_k),\,\forall \bu_k\in\Uset_k$, in subsequent iterations, does not drop below this threshold. Consequently, the final $\underset{\bu_k\in\Uset_k}{\min}\SNR_k(\bu_k)$ value approaches the threshold, owing to the reduction in $\underset{n}{\max}\,\,[\bt_k]_n$. Additionally, the maximum differential phases of LC-RISs can be increased beyond $2\pi$ by increasing the phase shifter length ($\ell_\lc$), see \eqref{eq: omega and kappa function static}. Typically, due to \MD{periodicity} of complex exponential function, $\omega_{\max}>2\pi$ does not improve the SNR; however, as can be seen from Fig.~\ref{fig:deltaomega}, $\omega_{\max}>2\pi$ can significantly improve the LC-RIS transition time due to \MD{the influence of $\omega_{\max}$ as} can be observed in \eqref{eq: time positive} and \eqref{eq: time negative}. When the length of the LC is increased to generate a $4\pi$ differential phase shift, the transition time decreases further, highlighting the efficiency improvement.

\subsection{Histogram of Differential Phase Shift Values}
In Fig.~\ref{fig: pdf_delta_omega}, we plot the histogram of the differential phase shifts before and after optimization for $K=3$ users. \MD{In particular, the horizontal axis represents the value of the differential phase shifts in radians, ranging from $-2\pi$ to $2\pi$. This reflects how much the phase values differ between successive configurations. The vertical axis represents the frequency (or normalized probability density) of these differential values, i.e., how often each range of differential phase values occurs. For the benchmark method, the phase-shift distribution is approximately uniform. Hence, the differential phase shifts span the full range and resemble a triangular-like distribution. This is expected from subtracting two uniform distributions (the convolution of two uniform distributions is a triangular-like distribution). This uniformity is observed because for a large \gls{NF} RIS serving a designated area, the optimal phase required to satisfy the \gls{SNR} constraint varies significantly across the surface, causing the phase values to be spread across the entire $(0,2\pi)$ range.

However, for the proposed method, the distribution becomes sharply peaked around zero, indicating that most differential shifts are small. This change reflects the impact of optimizing the \gls{RIS} phase profiles to reduce configuration switching time, as smaller phase differences between successive configurations result in faster reconfiguration. Hence, this figure visually confirms the connection between smaller differential phase shifts and faster configuration times.}
\begin{figure}
    \centering
    \includegraphics[width=1\linewidth]{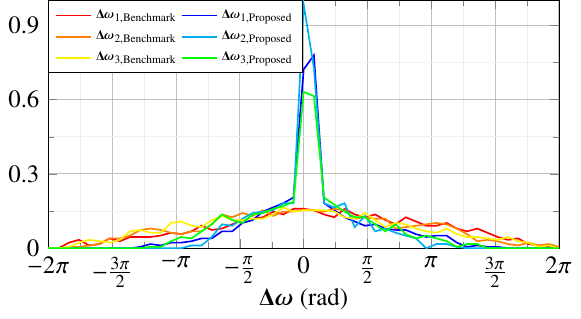}
    \caption{Histogram of $\bDelta\bomega$ in the proposed and benchmark designs for each configuration.}
    \label{fig: pdf_delta_omega}
     \vspace{-5mm}
\end{figure}

\subsection{Transition Time Reduction and Impact of Scheduling Order}
\label{sec: Transition Time Reduction and Impact of Scheduling Order}
Fig. \ref{fig:TDMA_time} shows the \gls{SNR}~(dB) when the \gls{RIS} configuration is switched every 57~ms to serve the users. The selected time interval is intended to showcase the proposed algorithm’s real-time performance.
\MDD{As is consistently illustrated in the Figs.~\ref{fig:TDMA_time_1_OPT}, \ref{fig:TDMA_time_2_OPT}, \ref{fig:TDMA_time_1}, and \ref{fig:TDMA_time_2}, our algorithm reaches the required \gls{SNR} threshold (i.e., 10 dB) significantly faster than the benchmarks (analytical and \gls{SDR}), which simply maximizes the minimum received \gls{SNR} in the targeted area without accounting for the transition time.}
\MD{In particular, our algorithm prioritizes minimizing transition time and is therefore designed to find a phase-shift configuration that satisfies the minimum SNR constraint in the targeted area as quickly as possible. Specifically, the algorithm does this by identifying a phase-shift configuration that satisfies the \gls{SNR} requirement while making adjustments to the existing \gls{RIS} phase-shifts.
\MDD{The benchmarks, conversely, only searches for the configuration that maximizes the final SNR, regardless of the time taken.}
For time-critical \gls{TDMA} applications, reaching the operational threshold quickly is an important performance metric that confirms the practical advantage of our approach.} \MDD{\gls{SDR} outperforms the analytical method by achieving a higher steady-state \gls{SNR}, but at the cost of higher complexity. Consequently, it is feasible to implement it with $N=300$. Moreover, the basis for the design of the low-complexity analytical design in \cite{Jamali2022lowtozero} is that the \gls{RIS} is extremely large. This is why the analytical benchmark performs poorly for small $N$ but well for large $N$.}

In addition, we analyze the impact of ordering users in accelerating the required time for phase configuration in RIS. In general, when $K$ users must be served in TDMA application, there are $(K-1)!$ different possibilities for ordering due to circularity. In the simplest case, let us assume $K=3$ and there are two different way to support users.
\MDD{As can be seen in Figs.~\ref{fig:TDMA_time_1_OPT}, \ref{fig:TDMA_time_2_OPT}, \ref{fig:TDMA_time_1}, and \ref{fig:TDMA_time_2}, the minimum SNR in the targeted area is plotted against the time in millisecond (ms). In all figures, there are improvements in time but in comparison of two set of configurations, the second configuration (User 1 - User 3 - User 2) can decrease more the required time compared to first configuration (User 1 - User 2 - User 3) showing the importance of ordering in \gls{LC}-RISs.}

\begin{figure*}
\begin{subfigure}{0.5\textwidth}
    \centering
    \includegraphics[width=1\textwidth]{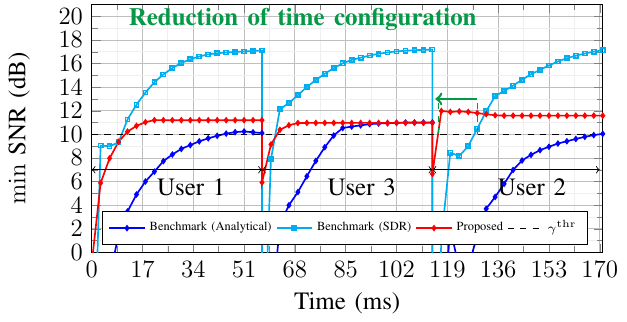}
    \caption{Scheduling order: User 1 - User 2 - User 3, and $N=300$.}
    \label{fig:TDMA_time_1_OPT}
\end{subfigure}
\begin{subfigure}{0.5\textwidth}
    \centering
    \includegraphics[width=1\textwidth]{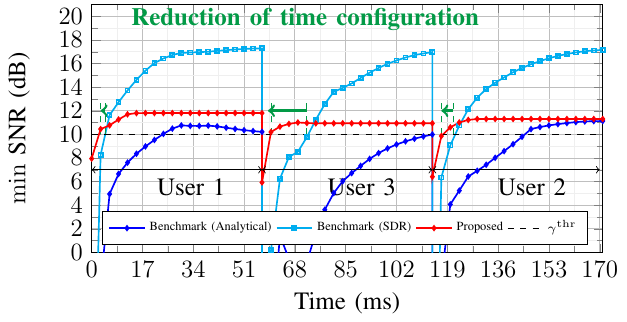}
    \caption{Scheduling order: User 1 - User 3 - User 2, and $N=300$.}
    \label{fig:TDMA_time_2_OPT}
\end{subfigure}
\begin{subfigure}{0.5\textwidth}
    \centering
    \includegraphics[width=1\textwidth]{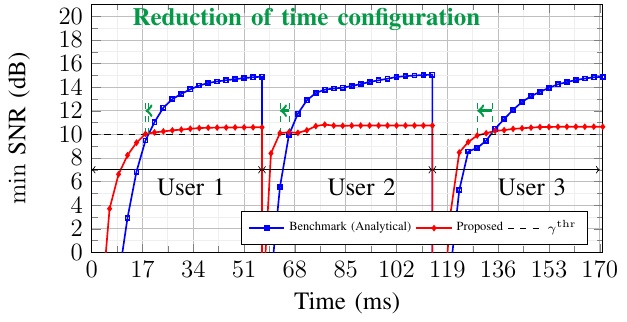}
    \caption{Scheduling order: User 1 - User 2 - User 3, and $N=700$.}
    \label{fig:TDMA_time_1}
\end{subfigure}
\begin{subfigure}{0.5\textwidth}
    \centering
    \includegraphics[width=1\textwidth]{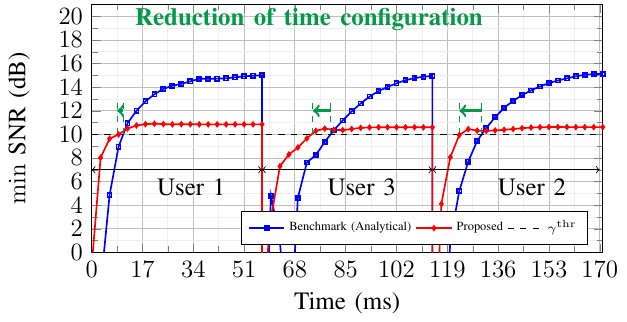}
    \caption{Scheduling order: User 1 - User 3 - User 2, and $N=700$.}
    \label{fig:TDMA_time_2}
\end{subfigure}
\caption{SNR comparison between benchmarks and proposed algorithm with 57 ms serving the users in two different cases of users' sequence.}
\label{fig:TDMA_time}
\end{figure*}

\subsection{Effective Data Rate}
In Fig.~\ref{fig:data rate}, we plot the effective data rate given by
\begin{equation}
  \label{eq:rate}
    R=\frac{\max(T_s-T_c,0)}{T_s}\log_\MD{2}{(1+\mathrm{SNR_{thr}})},
\end{equation}
where $T_c$ is the time RIS needs to reconfigure and reach $\SNR_\thr$ and $T_s$ is the time interval that RIS switches between serving the users (determined by the application scenario, e.g., delay requirement, users' mobility). Since there are multiple users, we plot the results based on the average $T_c$ across all users. When $T_s$ is smaller than this average value, the effective data rate $R$ is zero for both algorithms due to insufficient time for reconfiguration, as implied by \eqref{eq:rate}. As $T_s$ increases, the proposed algorithm achieves a higher data rate for the users. In the limit as $T_s \to \infty$, both algorithms converge to \( R \to \log_\MD{2}(1 + \SNR_\text{thr}) \). \MD{Furthermore, Fig.~\ref{fig:data rate} shows that increasing the maximum phase shift range from $2\pi$ to $4\pi$ significantly improves the effective data rate. This demonstrates a key insight of our work: even though a $2\pi$ phase range is sufficient for arbitrary beamforming at the steady state, a larger physical range ($\omega_{\max}$) directly accelerates the reconfiguration time, as modeled in \eqref{eq: time positive} and \eqref{eq: time negative}. A larger $\omega_{\max}$ provides a greater dynamic range for the phase transition, increasing the rate of phase change and thus reducing the time $T_c$ required to reach any desired intermediate phase. This highlights a crucial design trade-off where increasing the \gls{LC}-\gls{RIS} phase shifter length can be leveraged to overcome the inherent slow response time of the material.} Additionally, increasing the maximum phase shift range from $2\pi$ to $4\pi$ further improves the effective data rate in both algorithms.

\begin{figure}
    \centering
    \includegraphics[width=1\linewidth]{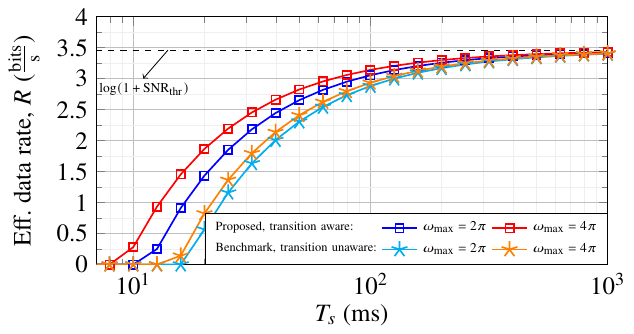}
    \caption{The average effective data rate ($R$) versus switching time between three users,
$T_s$, for two different algorithms.}
    \label{fig:data rate}
     \vspace{-5mm}
\end{figure}

\MD{
\subsection{Validation of the LoS-Dominant Link Assumption}
Figure~\ref{fig:CSI-LOS} shows the improvement in effective data rate ($R$) of the proposed transition-aware design (Algorithm~\ref{alg:cap}) over the transition-unaware benchmark as a function of the Rician $K$-factor. The red curve reports the gain when the \gls{RIS} knows only the \gls{LOS} component (derived from user location), whereas the blue curve represents an upper bound obtained with full \gls{CSI} including both \gls{LOS} and non-\gls{LOS} components. The gap to the upper bound decreases as the $K$-factor increases. For the parameters considered, the data rate achieved by the \gls{LOS}-based scheme approaches the full-\gls{CSI} performance for $K \gtrsim 3$ dB. Such $K$-factor values are typical in \gls{mmWave} environments~\cite{Mukherjee2017}, which supports the \gls{LOS}-based design adopted in this work.

Figure~\ref{fig:CSI-LOS} also confirms that Algorithm~\ref{alg:cap} is not limited to using only \gls{LOS} link information; it can also improve the effective data rate when full CSI knowledge is available. Specifically, our algorithm operates by taking an initial target phase-shift configuration, denoted as $\bW^{(0)}=[\bphi_1, \cdots, \bphi_K]$, which is assumed to be pre-computed based on the available \gls{CSI} (whether \gls{LOS}-based or full). Our algorithm then takes this $\bW^{(0)}$ as input and computes a new set of phase shifts $\bW=[\bomega_1, \cdots, \bomega_K]$ that satisfies the \gls{QoS} constraint while minimizing the transition time from the previous state.
}

\begin{figure}
    \centering
    \includegraphics[width=1\linewidth]{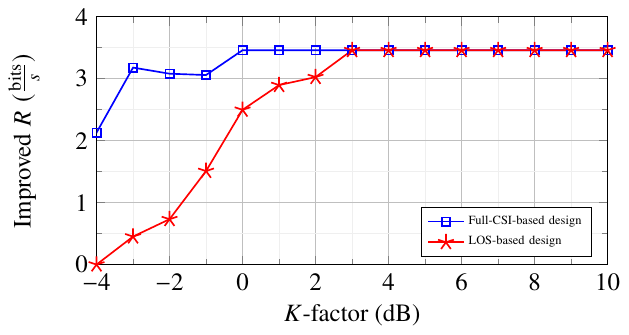}
    \caption{\MD{Improvement in effective data rate ($R$) of the proposed transition-aware design (Algorithm~\ref{alg:cap}) relative to a transition-unaware benchmark versus the Rician $K$-factor, for LoS-only knowledge and full CSI (LoS and NLoS).}}
    \label{fig:CSI-LOS}
\end{figure}

\section{Experimental Results}
\label{experimental results}
In this section, we evaluate the proposed algorithm using a small-scale proof of concept implementation of an \gls{LC}-\gls{RIS}. The experimental LC-RIS comprises  $30\times25 = 750$ elements, with the same voltage applied across all elements within a column for 1D beamforming\footnote{The reason is the experimental hardware limitations, otherwise, our proposed algorithm can also be used in 2D.}. This configuration allows the LC-RIS to reflect waves toward different azimuth angles but not varying elevation angles. These bias voltages are applied with a 1~kHz square wave. We used DAC60096 EVM with 12 bits from Texas Instruments providing us $\pm10.5$~V. As shown in Fig.~\ref{fig:LC_setup}, we used two 25dBi V-band horn antennas from MI-wave where the \gls{Tx} is fixed at an azimuth angle of $30^\circ$ and the distance of 
$1~$m relative to the LC-RIS, while the \gls{Rx} can rotate for any azimuth angle with a fixed distance of $55~$cm from the LC-RIS. For this example, the receiver is set to an angle of $-30^\circ$. The measurements have been performed with a PNA-X N5247A from Keysight Technologies. 

\begin{figure}
    \centering
    \includegraphics[width=0.35\textwidth]{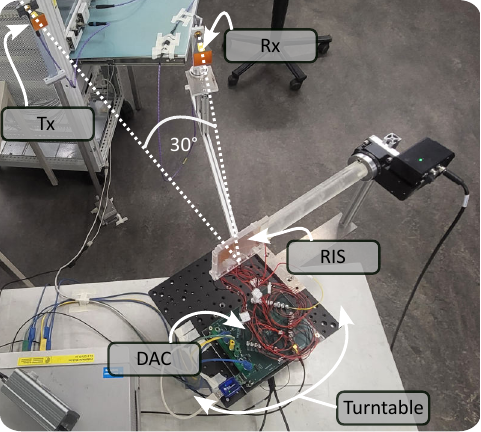}
    \caption{Experimental setup for measuring the \gls{SNR} in different positions. For more details of the setup, you can see \cite[Fig.~5]{neuder2023architecture}.}
    \label{fig:LC_setup}
    \vspace{-5mm}
\end{figure}

We evaluate two types of phase-shifter optimization strategies: benchmark and the proposed method. 
\MD{As our work focuses in addressing the dynamic reconfiguration latency of an \gls{LC}-\gls{RIS}, we consider time reduction in satisfying the \gls{QoS} as the performance metric same as Section~\ref{sec: Transition Time Reduction and Impact of Scheduling Order}. We evaluate our method against the most powerful alternative: a benchmark algorithm that is transition-unaware but is otherwise optimal for the static case, meaning it only tries to maximize the final \gls{SNR} without accounting for the transition time \cite{najafi2020physics,Alexandropoulos2022,Lu2024}.}
Consider a scenario with two receivers located at $\theta_1=-30^\circ$ and $\theta_2=30^\circ$, and a transmitter located at $\theta=0^\circ$. During a transition, the LC-RIS must shift its focus from serving the first receiver to the second. This requires a seamless adjustment of phase shifters to ensure fast transition. With linear phase shift design, as can be seen from Fig.~\ref{fig:experimental_benchmark}, it takes about $148$~ms to received power reaches the required $\SNR_\thr$ which is $10$~dB higher than noise power (note that the fabrication, i.e., the rubbing, of the alignment layer is not ideal which results in a longer response time compared to the reported values in Fig.~\ref{fig:time response+-}). However, with the proposed phase shifter design, Fig. \ref{fig:experimental_proposed}, the required time reduces to only $60$~ms at the cost of the reducing in final received SNR. These results suggest the effectiveness of the proposed modeling and algorithm design in practice. \MD{It is important to note that the starting \gls{SNR} for the second user is higher in the proposed method (Fig.~\ref{fig:experimental_benchmark}) compared to the benchmark (Fig.~\ref{fig:experimental_proposed}). This is not an artifact of the experiment but a direct result of the transition-aware algorithm. The benchmark, being unaware of the subsequent user, generates a configuration for the first user that is solely optimal for that user's link. In contrast, our algorithm finds a configuration for the first user that not only satisfies its \gls{SNR} requirement but is also 'closer' to the optimal configuration for the second user, thereby pre-conditioning the \gls{RIS} for a faster transition. Consequently, the higher initial \gls{SNR} is an integral part of the overall performance gain demonstrated by our approach.} The final received SNR for benchmark and the proposed algorithms are plotted in Figs.~\ref{fig:heatmap_benchmark} and \ref{fig:heatmap_proposed}, respectively. Although the setup is designed for $f = 60$~GHz, the received power spans a frequency range of $53-67$~GHz. Outside $f=60$~GHz, beam splitting effects are observed, which may be considered as a potential direction for future research.

\begin{figure*}
\begin{subfigure}{0.49\textwidth}
    \centering
    \includegraphics[width=1\textwidth]{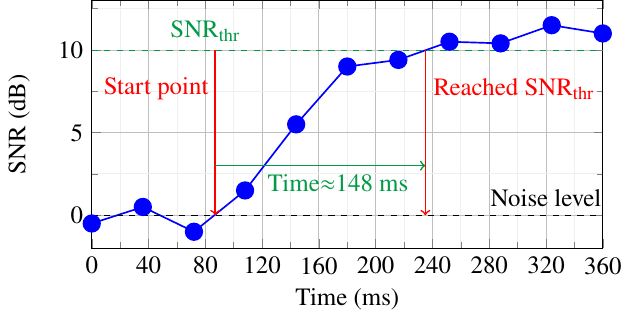}
    \caption{Measurement results for the benchmark phase shift optimization, which requires approximately 148~ms to achieve the $\SNR_\thr$. This algorithm only focuses on maximizing the final receive power.}
    \label{fig:experimental_benchmark}
\end{subfigure}
\hfill
\begin{subfigure}{0.49\textwidth}
    \centering
    \includegraphics[width=1\textwidth]{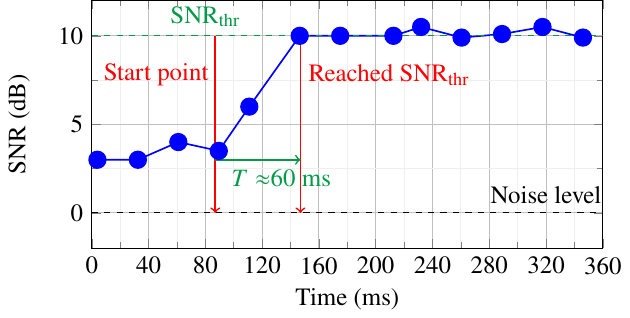}
    \caption{Measurement results for the proposed phase shift optimization, which requires approximately 60~ms to achieve the $\SNR_\thr$. This algorithm minimizes the transition time while reaching the $\SNR_\thr$.}
    \label{fig:experimental_proposed}
\end{subfigure}

\begin{subfigure}{0.48\textwidth}
    \centering
    \includegraphics[width=0.7\textwidth]{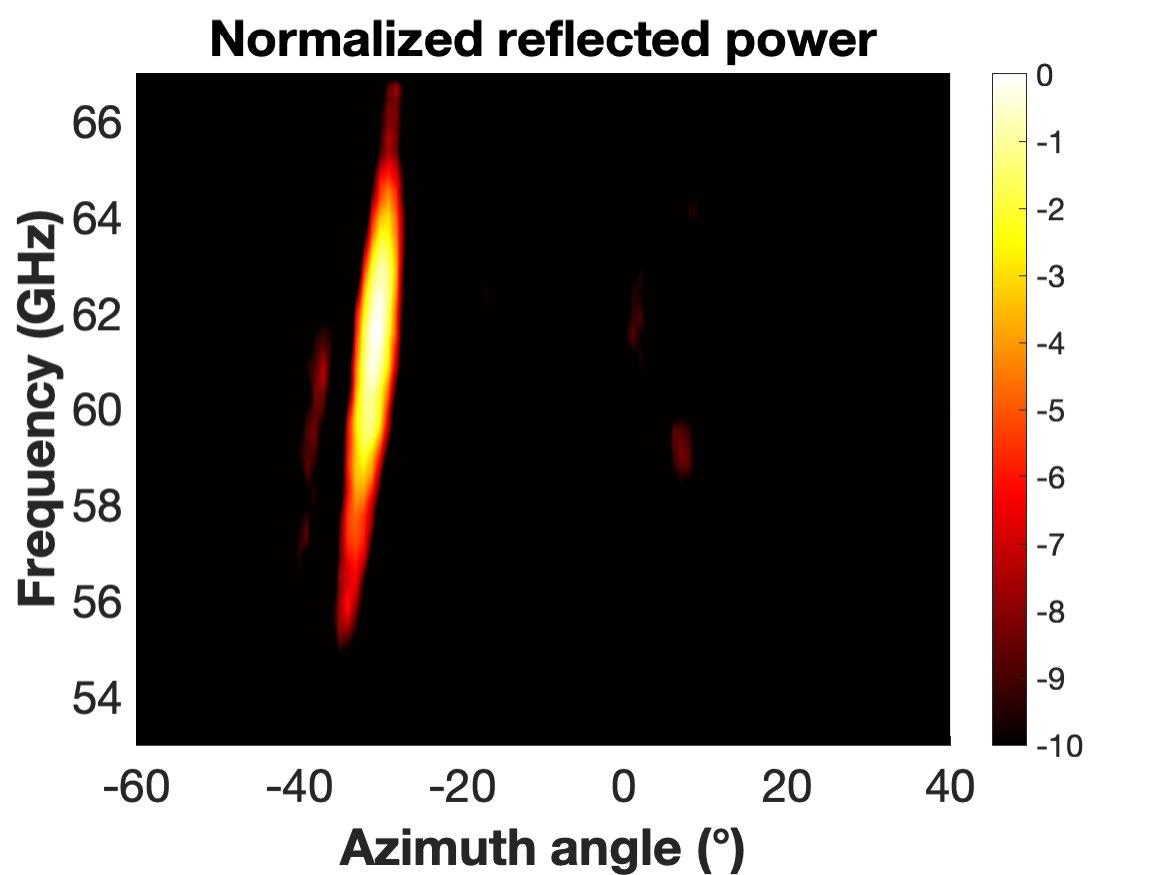}
    \caption{Measured received signal power after sufficient time. The benchmark algorithm maximizes the final received power regardless of the transition time.}
    \label{fig:heatmap_benchmark}
\end{subfigure}
\hfill
\begin{subfigure}{0.48\textwidth}
    \centering
    \includegraphics[width=0.7\textwidth]{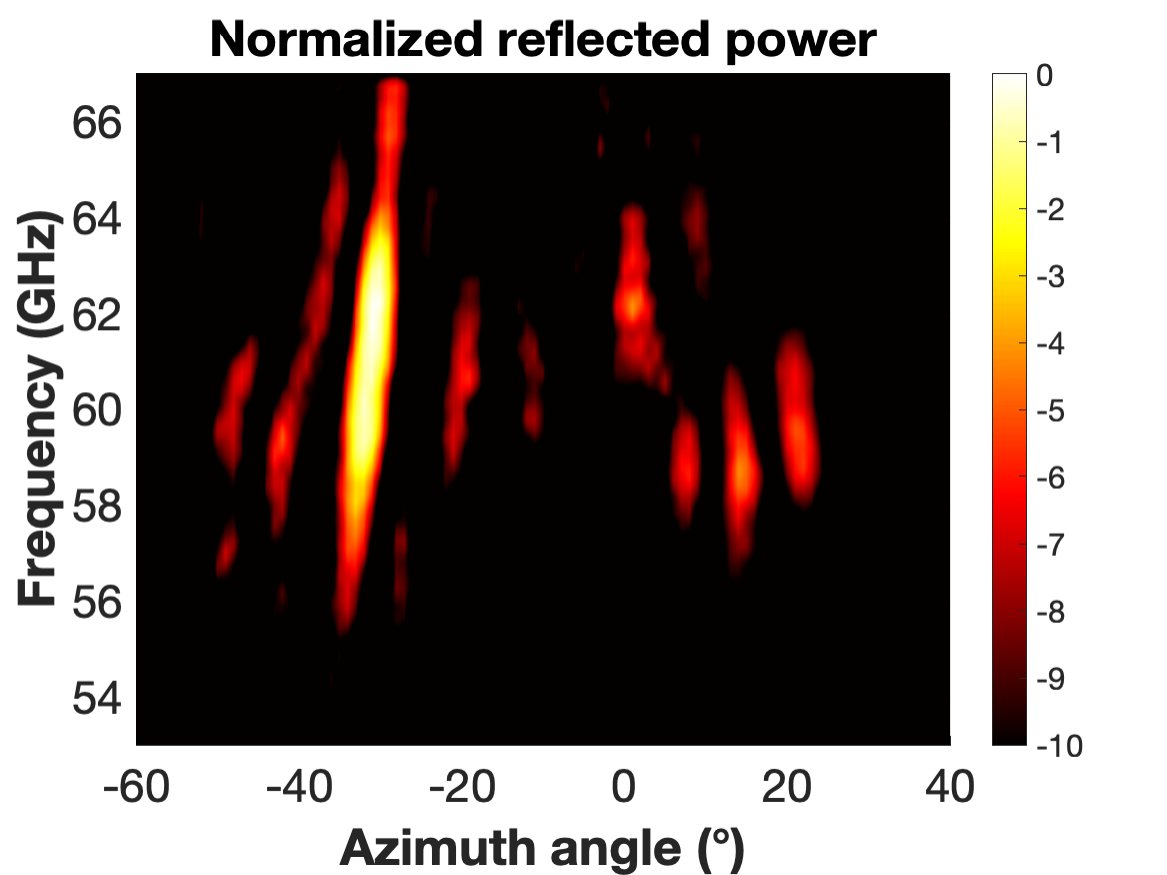}
    \caption{Measured received signal power after sufficient time. The proposed algorithm accelerates the transition time at the cost of a slight reduction in the final received power; however, it remains well above the required $\SNR_\thr$.}
    \label{fig:heatmap_proposed}
\end{subfigure}
\caption{A comparison between the two algorithms is conducted when applying new phase shifts on the LC-RIS to serve the second user. In this analysis, $\SNR_\thr$  is set to be 10~dB higher than the noise power.}
\vspace{-5mm}
\end{figure*}
\section{Conclusions and \MD{Future Directions}}
\label{conclusion}
In this paper, we have studied \gls{LC}-\glspl{RIS} as a scalable cost-effective solution to realize large \glspl{RIS} required at mmWave band to generate sufficient link budget. The primary challenge of \gls{LC}-\glspl{RIS} is however their slow response. To address this challenge, we have first developed a physics-driven model to characterize the response time of LC-RIS. Based on this model, we proposed an efficient algorithm for fast reconfiguration of LC-RISs to serve $K$ users in a \gls{TDMA} framework.  Simulation results validated the effectiveness of the proposed transition-aware design, showcasing its significant performance advantages over existing transition-unaware benchmarks in the literature.
\MDD{Specifically, our analysis and results revealed that the proposed algorithm reduces the transition time by selecting phase configurations that minimize the differential phase shift required between consecutively served \glspl{MU}. This approach effectively preconditions the \gls{RIS} state, leading to a higher starting \gls{SNR} for the subsequent user and a smaller transition range, thereby satisfying the \gls{QoS} requirement significantly faster than conventional transition-unaware designs.}
Furthermore, the experimental validation confirmed the effectiveness of the algorithm in a proof of concept implementation of LC-RIS in an indoor application.

\MD{This paper lays the necessary foundation for several extensions to fully investigate the potential of \gls{LC}-\glspl{RIS}. Although our study focused on the TDMA protocol, extending the transition-aware design to multi-user scenarios with hybrid beamforming at the BS requires a joint optimization of interference mitigation and reconfiguration latency, which constitutes a promising research direction \cite{Lu2025}. Furthermore, while the proposed algorithm is inherently robust to location estimation errors, it has been primarily developed for LOS-dominated environments. In scenarios where non-LOS links cannot be neglected, robust designs that account for imperfect full CSI are required. Finally, while our simulations highlighted the impact of scheduling order and LC phase-shifter length on performance, a more comprehensive analysis of these factors is left for future work.}


\appendices
\section{Proof of Lemma \ref{lem: PDE}}
\label{app: PDE proof}
Lemma~\ref{lem: PDE} deals with solving a modified heat \gls{PDE} that incorporates a nonlinear external source $\varepsilon_0\Delta\varepsilon E^2\Phi(\varphi)$. In the following, we solve this PDE by analyzing two specific cases:
\begin{itemize}
    \item The electric field decreases from its maximum value ($E_{\max}$) to zero (decaying in time).
    \item The electric field increases from zero to $E_{\max}$ (rising in time).
\end{itemize}

\textbf{Case 1: Decaying in time:} When $E=0$, the nonlinear term $\varepsilon_0\Delta\varepsilon E^2\Phi(\varphi)$ in \eqref{eq: elastic 1D simple} vanishes, reducing the \gls{PDE} to a standard heat equation:
\begin{subequations}
\begin{align}
        \label{eq: PDE 1D simple Heat equation}
&\gamma_1\frac{\partial\varphi}{\partial t}=\bar{K}\frac{\partial^2 \varphi}{\partial z^2}
    \\&~\text {C1, C2, and C3}.
\end{align}
\end{subequations}
Using the method of separation of variables \cite{kreyszig}, the general solution of \eqref{eq: PDE 1D simple Heat equation} subject to C1 and C2 is:
\begin{equation}
\label{eq: decaying time}
    \varphi(z,t)=\sum_{n=1}^\infty A_n\sin(\frac{n\pi z}{d})\e^{-\frac{\bar{K}}{\gamma_1}(\frac{n\pi}{d})^2t},
\end{equation}
where the coefficients $A_n$ are determined by the initial condition:
\begin{equation}
    A_n=\frac{2}{d}\int_0^d\varphi(z,0)\sin(\frac{n\pi z}{d})\dd z.
\end{equation}
Due to the initial condition in C3: $\varphi(z,0)=A_1\sin(\frac{\pi z}{d})$, we find that $A_n = 0$ for all $n > 1$. Hence, the solution simplifies to:
\begin{equation}
\label{eq: decaying time simplified}
    \varphi(z,t)=0+ (A_1-0)\sin(\frac{\pi z}{d})\e^{-\frac{t}{\tau_\mol^-}},
\end{equation}
where $\tau_\mol^-=\frac{\gamma_1}{\bar{K}}(\frac{d}{\pi})^2$.\\

\textbf{Case 2: Rising in time:} 
When $E=E_{\max}$, the nonlinear term 
$\varepsilon_0\Delta\varepsilon E^2\Phi(\varphi)$  becomes significant. This \gls{PDE} has been solved in \cite{wang2004correlations,Wang2005} with the following solution:
\begin{equation}
\label{eq: arising time}
    \varphi(z,t)= \frac{\pi}{2}\sqrt{\frac{1}{1+(\frac{\pi^2}{4A_1^2}-1)\exp(-\frac{t}{\tau_r})}}\sin(\frac{\pi z}{d}),
\end{equation}
where $\tau_r=\frac{\gamma_1}{2|\varepsilon_0\Delta\varepsilon E_{\max}^2-\frac{\pi^2}{d^2}\bar{K}|}\propto E_{\max}^{-2}$\footnote{$E_t\defeq\sqrt{\frac{\bar{K}}{\varepsilon_0\Delta\varepsilon}}\frac{\pi}{d}$ represents the Freedericksz threshold field necessary to initiate rotation of the LC molecules.} (recall $E_{\max}\gg\frac{\pi^2\bar{K}}{d^2\varepsilon_0\Delta\varepsilon}$). We now prove that the following approximation instead of \eqref{eq: arising time} has an error bounded by $\bigO(E_{\max}^{-2})$:
\begin{equation}
\label{eq: arising time final}
        \varphi(z,t)=\left(\frac{\pi}{2}-(\frac{\pi}{2}-A_1)\e^{-\frac{t}{\tau_\mol^+}}\right)\sin(\frac{\pi z}{d}),
\end{equation}
where $\tau_\mol^+=\alpha\tau_r$, and $\alpha\geq1$ is a tuning parameter (in this case, we know $\varphi(z,\infty)=\frac{\pi}{2}\sin(\frac{\pi z}{d})$ after applying $E=E_{\max}$ \cite{wang2004correlations,Wang2005}).

\textit{Error Analysis:} To analyze the error, we compute the integral of the absolute difference between \eqref{eq: arising time} and \eqref{eq: arising time final} over time. For an upper bound, we consider $z=\frac{d}{2}$, and $\alpha=1$:
\begin{align}
    \int_0^\infty &\frac{\pi}{2}\left(1-\sqrt{\frac{1}{1+(\frac{\pi^2}{4A_1^2}-1)\exp(-\frac{t}{\tau_\mol^+})}}\right)\nonumber\\
    +&(\frac{\pi}{2}-A_1)\e^{-\frac{t}{\tau_\mol^+}} \dd t.
\end{align}
Evaluating this integral yields:
\begin{equation}
    \tau_\mol^+\left(\underbrace{\log\left(\frac{\frac{\pi}{2A_1}+1}{\frac{\pi}{2A_1}-1}\right)+\log\left(\frac{\frac{\pi^2}{4A_1^2}-1}{4}\right)+\frac{\pi}{2}-A_1}_{\beta}\right).
\end{equation}
Hence, the error is bounded by $\tau_\mol^+\beta$, where $\tau_\mol^+\propto E_{\max}^{-2}$. This confirms that the approximation in \eqref{eq: arising time final} holds with an error of $\bigO(E_{\max}^{-2})$. This concludes the proof.

\section{Proof of Proposition~\ref{prop: linearity in time}}
\label{app: linearity in time}
To characterize the time-dependent behavior of the LC phase shifter, we begin by considering \eqref{eq: omega and kappa function static}. In this equation, \(\ell_\lc\) and \(\kk_\perp\) remain constant over time, while \(\kk(t)\) is time-dependent. In general, \(\kk(t)\) is given by:
\begin{equation}
\label{eq: kappa in time}
\kk(t)=\frac{1}{d}\int_0^d \left(\left(\frac{l(\varphi(z,t))}{a}\right)\kk_\lc+\left(\frac{a-l(\varphi(z,t))}{a}\right)\kk_0\right) \dd z,
\end{equation}
where $a$ represents the semi-major axis of the ellipsoidal shape LC, $l(\varphi(z,t))$ is the effective length of the LC region contributing to the phase shift, $\kk_\lc$ and $\kk_0$ are the wave numbers in LC material and vacuum, respectively, see Fig.~\ref{fig:lc-ab}. As shown in the left figure of Fig.~\ref{fig:lc-ab}, to compute \(\kk_\perp\), we substitute $l(\varphi(z,t)) = b,\,\forall z$ in \eqref{eq: kappa in time}, where $b$  is the semi-minor axis of the LC ellipsoid. Then, we can express the dynamic phase shift function in \eqref{eq: omega and kappa function static} as:
\begin{equation}
    \label{eq: omega and time}
    \omega(t)=\ell_\lc\frac{1}{d}\int_0^d \left(\left(\frac{l(\varphi(z,t))-b}{a}\right)\left(\kk_\lc-\kk_0\right)\right) \dd z.
\end{equation}
In the following, we first derive the function of $l(\varphi)$. Then, we establish that the phase shift generated by each \MD{\gls{LC}-\gls{RIS} phase shifter} evolves as a summation on exponential functions of time. We begin by proving that the effective length of an LC molecule, depicted in Fig. \ref{fig:lc-ab}, is given by:
\begin{equation}
    l(\varphi)=b\sqrt{\frac{1}{1+(\frac{b^2}{a^2}-1)\sin^2(\varphi)}},
\end{equation}
\begin{figure}
    \centering
    \includegraphics[width=0.5\textwidth]{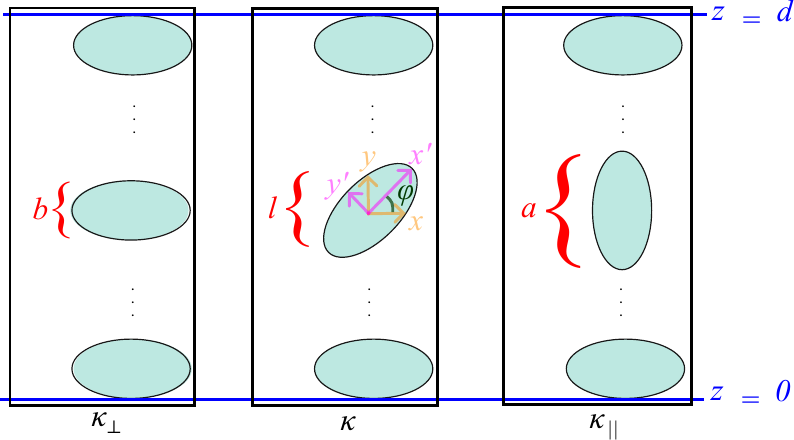}
    \caption{A rotated ellipse shape \gls{LC} molecule with semi-major and semi-minor diameters $a$ and $b$, respectively.}
    \label{fig:lc-ab}
    \vspace{-5mm}
\end{figure}

\textbf{Derivation of $l(\varphi)$:} The result can be derived using the following coordinate transformation:
\begin{equation*}
\begin{bmatrix}
x\\
y
\end{bmatrix}
=\begin{bmatrix}
\cos(\varphi) & \sin(\varphi)\\
-\sin(\varphi) & \cos(\varphi)
\end{bmatrix}
\begin{bmatrix}
x'\\
y'
\end{bmatrix}.
\end{equation*}
From Fig. \ref{fig:lc-ab}, we know that the effective length is $l=2|y|\big|_{x=0}$. Substituting $x=0$ into the transformation equations, we obtain:
\begin{subequations}
\begin{equation}
    y'=-\cot(\varphi)x'
\end{equation}
\begin{equation}
    y=(-x'\sin(\varphi)+y'\cos(\varphi)).
\end{equation}
\end{subequations}
Combining these two equations yields:
\begin{equation}
    x'=\frac{-y\sin(\varphi)}{2},\,\quad y'=\frac{y\cos(\varphi)}{2}.
\end{equation}
Next, we substitute these expressions into the equation of the ellipse:
\begin{equation}
    \frac{4x'^2}{a^2}+\frac{4y'^2}{b^2}=1.
\end{equation}
After substitution and simplification, we solve for $y$ and find:
\begin{equation}
    y=\pm\frac{ab}{2\sqrt{a^2\cos^2(\varphi)+b^2\sin^2(\varphi)}}.
\end{equation}
Since $l=2|y|$, the effective length becomes:
\begin{equation}
\label{eq: function l(varphi)}
    l=b\sqrt{\frac{1}{1+(\frac{b^2}{a^2}-1)\sin^2\varphi}}.
\end{equation}
This concludes the derivation of $l(\varphi)$.

\textbf{Expansion of the Phase Shift $\omega(t)$:} Using the result from the first part and Lemma~\ref{lem: PDE}, we expand $l(\varphi(z,t))$ around $\varphi(z,\infty)$ by exploiting Taylor series as:
\begin{equation}
\label{eq: l and varphi}
l(\varphi(z,t))=b(1+\sum_{p=0}^{\infty}C_p(\varphi(z,t)-\varphi(z,\infty))^{p}).
\end{equation}
where $C_p$ are the coefficients of the expansion.
Firstly, substitute \eqref{eq: varphi final} in \eqref{eq: l and varphi}. Secondly, substitute the result into \eqref{eq: omega and time} and interchange the summation and integral, we get:
\begin{equation}
\label{eq: sum integral}
    \omega(t)=\ell_\lc\sum_{p=0}^{\infty} \e^{-\frac{pt}{\tau_\mol}} \underbrace{\frac{C_p}{ad}\int_0^d (\varphi(z,0)-\varphi(z,\infty))^{p}(\kappa_{\lc}-\kappa) \dd z}_{\text{A function of $p$}}.
\end{equation}
This finally yields to:
\begin{equation}
\label{eq: sum integral final}
    \omega(t)=\ell_\lc\sum_{p=0}^{\infty}D_p\e^{-\frac{pt}{\tau_\mol}},
\end{equation}
where $D_p$ are constant coefficients. This concludes the proof.

\section{Proof of Lemma~\ref{theorem: fourier}}
\label{app: theorem fourier}
    By substituting the phase shifter vector ($\bv$) in \gls{SNR} definition \eqref{eq: SNR new}, we have
    \begin{equation}
    \begin{array}{ll}
        &\SNR_k(\bu_k^{(p)})=\sum_{n=1}^N\sum_{m=1}^N\e^{-\jj[\bomega_k]_n}[\bbM_k(\bu_k^{(p)})]_{n,m}\e^{\jj[\bomega_k]_m}.
    \end{array}
    \end{equation}
    Because \gls{LOS} link is dominant on both BS-RIS and RIS-MU, we can decompose $[\bbM_k(\bu_k^{(p)})]_{n,m}=\e^{[\bphi_k(\bu_k^{(p)})]_{n}}\e^{-[\bphi_k(\bu_k^{(p)})]_{m}}$ where phase shifter $[\bphi_k(\bu_k^{(n)})]_{p},\,\forall n$ can be derived by \eqref{Eq:LoSnear} \cite{delbari2024nearfield}. Let us expand SNR around the phase shifter $[\bomega]_p$ where all other phase shifters for any $m\neq p$ are fixed. So
    \begin{equation}
        \begin{array}{ll}
        \label{eq: SNR final final}
            &\SNR_k(\bu_k^{(p)})=\e^{[\bomega]_p-[\bphi_k(\bu_k^{(p)})]_{p}}\sum_{m\neq p}\e^{-[\bomega]_m+[\bphi_k(\bu_k^{(p)})]_{m}}\\
            &+\e^{-[\bomega]_p+[\bphi_k(\bu_k^{(p)})]_{p}}\sum_{m\neq p}\e^{[\bomega]_m-[\bphi_k(\bu_k^{(p)})]_{m}}\\
            &+1+\sum_{m\neq p}\e^{[\bomega]_m-[\bphi_k(\bu_k^{(p)})]_{m}}\sum_{m\neq p}\e^{-[\bomega]_m+[\bphi_k(\bu_k^{(p)})]_{m}}.
        \end{array}
        \end{equation}
Let us define $[r_k]_p\e^{[\theta_k]_p}\defeq\sum_{m\neq p}\e^{[\bomega]_m-[\bphi_k(\bu_k^{(p)})]_{m}}$, where $[r_k]_p$ and $[\theta_k]_p$ are amplitude and phase, respectively. For sufficiently large $N$, $[\theta_k]_p$ tends to be zero $\forall p$ \footnote{This result can be derived by evaluating the integral over $[0, 2\pi]$, $\forall m\neq p$. However, due to space constraints, the detailed derivation is omitted}. Therefore, \eqref{eq: SNR final final} can be decomposed as $c+[r_k]_p\cos([\bomega_k]_p-[\phi_k(\bu_k^{(n)})]_p)$ where $c$ is a constant for $k$th user. Substituting this result into \eqref{eq: lagrangian user k over n} gives us:
\begin{align}
        [L_k]=&C_k+[\bxi_{k+1}]_p[\bt_{k+1}]_p+[\bxi_k]_p[\bt_k]_p\nonumber\\
    &-\underbrace{[\blambda_k]_p[r_k(\bu_k^{(n)})]_p}_{\text{New }[\blambda_k]_p}\cos([\bomega_k]_p-\bphi_k(\bu_k^{(p)})).
\end{align}
This concludes the proof.

\bibliographystyle{IEEEtran}
\bibliography{References}

@article{di2019smart,
	title={Smart Radio Environments Empowered by {AI} Reconfigurable Meta-Surfaces: {An} Idea Whose Time Has Come},
	author={Di Renzo, M.  and others},
	volume = {129},
        journal = {EURASIP J. Wireless Commun. and Netw.},
	year={2019},
	month = {May}
}

@ARTICLE{yu2021smart,
  author={Yu, Xianghao and Jamali, Vahid and Xu, Dongfang and Ng, Derrick Wing Kwan and Schober, Robert},
  journal={IEEE Wireless Commun.}, 
  title={Smart and Reconfigurable Wireless Communications: From {IRS} Modeling to Algorithm Design}, 
  year={2021},
  volume={28},
  number={6},
  pages={118-125},
  doi={10.1109/MWC.001.2100145}}

@article{najafi2020physics,
  title={Physics-based modeling and scalable optimization of large intelligent reflecting surfaces},
  author={Najafi, Marzieh and Jamali, Vahid and Schober, Robert and Poor, H Vincent},
  journal={IEEE Trans. Commun.},
  volume={69},
  number={4},
  pages={2673--2691},
  year={2020},
  publisher={IEEE}
}

@ARTICLE{qingqing2019IRS,
  author={Wu, Qingqing and Zhang, Rui},
  journal={IEEE Trans. Wireless Commun.}, 
  title={Intelligent Reflecting Surface Enhanced Wireless Network via Joint Active and Passive Beamforming}, 
  year={2019},
  volume={18},
  number={11},
  pages={5394-5409},
  doi={10.1109/TWC.2019.2936025}}

@article{bjornson2020rayleigh,
	title={Rayleigh fading modeling and channel hardening for reconfigurable intelligent surfaces},
	author={Bj{\"o}rnson, Emil and Sanguinetti, Luca},
	journal={Wireless Commun. Lett.},
	volume={10},
	number={4},
	pages={830--834},
	year={2020},
	publisher={IEEE}
}

@article{aboagye2022design,
  title={Design and Optimization of Liquid Crystal {RIS}-Based Visible Light Communication Receivers},
  author={Aboagye, Sylvester and Ndjiongue, Alain R and Ngatched, Telex MN and Dobre, Octavia A},
  journal={IEEE Photonics J.},
  volume={14},
  number={6},
  pages={1--7},
  year={2022},
  publisher={IEEE}
}

@article{ghannam2021reconfigurable,
  title={Reconfigurable surfaces using fringing electric fields from nanostructured electrodes in nematic liquid crystals},
  author={Ghannam, Rami and Xia, Yuanjie and Shen, Dezhi and Fernandez, F Anibal and Heidari, Hadi and Roy, Vellasaimy AL},
  journal={Advanced Theory and Simulations},
  volume={4},
  number={7},
  pages={2100058},
  year={2021},
  publisher={Wiley Online Library}
}

@article{guirado2022mm,
  title={{mm-Wave} Metasurface Unit Cells Achieving Millisecond Response Through Polymer Network Liquid Crystals},
  author={Guirado, Robert and Perez-Palomino, Gerardo and Ca{\~n}o-Garc{\'\i}a, Manuel and Geday, Morten A and Carrasco, Eduardo},
  journal={IEEE Access},
  volume={10},
  pages={127928--127938},
  year={2022},
  publisher={IEEE}
}

@article{garbovskiy2012liquid,
  title={Liquid crystal phase shifters at millimeter wave frequencies},
  author={Garbovskiy and others},
  journal={J. Applied Physics},
  volume={111},
  number={5},
  pages={054504},
  year={2012},
  publisher={American Institute of Physics}
}

@article{jakeman1972electro,
  title={Electro-optic response times in liquid crystals},
  author={Jakeman, E and Raynes, EP},
  journal={Physics Lett. A},
  volume={39},
  number={1},
  pages={69--70},
  year={1972},
  publisher={Elsevier}
}

@article{ericksen1961conservation,
  title={Conservation laws for liquid crystals},
  author={Ericksen, Jerald L},
  journal={Trans. of the Society of Rheology},
  volume={5},
  number={1},
  pages={23--34},
  year={1961},
  publisher={The Society of Rheology}
}

@article{leslie1968some,
  title={Some constitutive equations for liquid crystals},
  author={Leslie, Frank M},
  journal={Archive for Rational Mechanics and Analysis},
  volume={28},
  pages={265--283},
  year={1968},
  publisher={Springer}
}

@article{neuder2023architecture,
  title={Architecture for sub-100 ms Liquid Crystal Reconfigurable Intelligent Surface Based on Defected Delay Lines},
  author={Neuder, Robin and Sp{\"a}th, Marc and Sch{\"u}{\ss}ler, Martin and Jim{\'e}nez-S{\'a}ez, Alejandro},
    journal = {Commun. Eng.},
volume={3},
  number={1},
  pages={70},
  year={2024},
doi={https://doi.org/10.1038/s44172-024-00214-3}
}

@article{tse2005fundamentals,
  title={Fundamentals of Wireless Communication},
  author={Tse, D},
  journal={Cambridge University Press},
  volume={2},
  pages={614--624},
  year={2005}
}

@article{ndjiongue2021re,
  title={Re-configurable intelligent surface-based {VLC} receivers using tunable liquid-crystals: The concept},
  author={Ndjiongue, Alain R and Ngatched, Telex MN and Dobre, Octavia A and Haas, Harald},
  journal={J. Lightwave Technol.},
  volume={39},
  number={10},
  pages={3193--3200},
  month={May~15},
  year={2021},
  publisher={IEEE}
}

@INPROCEEDINGS{delbari2024fast,
  author={Delbari, Mohamadreza and Neuder, Robin and Jiménez-Sáez, Alejandro and Asadi, Arash and Jamali, Vahid},
  booktitle={Proc. IEEE International Conf. Commun. Workshops (ICC Workshops)}, 
  title={Fast Transition-Aware Reconfiguration of Liquid Crystal-Based {RIS}s}, 
  year={2024},
  address={Denver, CO, USA},
  volume={},
  number={},
  pages={214-219},
  doi={10.1109/ICCWorkshops59551.2024.10615422}}

@article{tawk2012varactor,
  title={A varactor-based reconfigurable filtenna},
  author={Tawk, Y and Costantine, J and Christodoulou, CG},
  journal={IEEE Antennas and Wireless Propag. lett.},
  volume={11},
  pages={716--719},
  year={2012},
  publisher={IEEE}
}

@article{tang2020wireless,
  title={Wireless communications with reconfigurable intelligent surface: Path loss modeling and experimental measurement},
  author={Tang, Wankai and Chen, Ming Zheng and Chen, Xiangyu and Dai, Jun Yan and Han, Yu and Di Renzo, Marco and Zeng, Yong and Jin, Shi and Cheng, Qiang and Cui, Tie Jun},
  journal={IEEE Trans. Wireless Commun.},
  volume={20},
  number={1},
  pages={421--439},
  year={2020},
  publisher={IEEE}
}

@book{ferrari2022reconfigurable,
  title={Reconfigurable circuits and technologies for smart millimeter-wave systems},
  author={Ferrari, Philippe},
  year={2022},
  publisher={Cambridge University Press}
}

@inproceedings{schmitt20213,
  title={3-Bit digital-to-analog converter with mechanical amplifier for binary encoded large displacements},
  author={Schmitt, Lisa and Schmitt, Philip and Hoffmann, Martin},
  booktitle={Actuators},
  volume={10},
  number={8},
  pages={182},
  year={2021},
}

@inproceedings{neuder2023compact,
  title={Compact Liquid Crystal-based Defective Ground Structure Phase Shifter for Reconfigurable Intelligent Surfaces},
  author={Neuder, Robin and Wang, Dongwei and Jakoby, Rolf and Jiménez-Sáez, Alejandro},
  booktitle={Proc. European Conf. Antennas and Propag. (EuCAP)},
address={Florence, Italy},
pages={1-5},
  year={2023},
 doi={10.23919/EuCAP57121.2023.10133014}
}

@article{jimenez2023reconfigurable,
  title={Reconfigurable intelligent surfaces with liquid crystal technology: A hardware design and communication perspective},
  author={Jim{\'e}nez-S{\'a}ez, Alejandro and Asadi, Arash and Neuder, Robin and Delbari, Mohamadreza and Jamali, Vahid},
  journal={arXiv preprint arXiv:2308.03065},
  year={2023}
}

@article{delbari2024far,
  title={Far-versus Near-Field {RIS} Modeling and Beam Design},
  author={Delbari, Mohamadreza and Alexandropoulos, George C and Schober, Robert and Jamali, Vahid},
  journal={arXiv preprint arXiv:2401.08237},
  year={2024}
}

@phdthesis{Wang2005,
  title={Studies Of Liquid Crystal Response Time},
  author={Haiying Wang},
  year={2005},
    school       = {Department of Electrical and Computer Engineering, University of Central Florida},
  address      = {Orlando, FL, USA},
  url={https://stars.library.ucf.edu/etd/632 }
}

@article{wang2004correlations,
  title={Correlations between liquid crystal director reorientation and optical response time of a homeotropic cell},
  author={Wang, Haiying and Wu, Thomas X and Zhu, Xinyu and Wu, Shin-Tson},
  journal={J. Applied Physics},
  volume={95},
  number={10},
  pages={5502--5508},
  year={2004},
  publisher={American Institute of Physics}
}

@misc{cvx,
  author       = {Michael Grant and Stephen Boyd},
  title        = {{CVX}: Matlab Software for Disciplined Convex Programming, version 2.1},
  howpublished = {\url{https://cvxr.com/cvx}},
  month        = mar,
  year         = 2014
}

@ARTICLE{Kim2023Independtly,
  author={Kim, Hogyeom and Oh, Seongwoog and Bang, Seungwoo and Yang, Hyunjun and Kim, Byeongjin and Oh, Jungsuek},
  journal={IEEE Trans. Antennas and Propag.}, 
  title={Independently Polarization Manipulable Liquid-Crystal-Based Reflective Metasurface for {5G} Reflectarray and Reconfigurable Intelligent Surface}, 
  year={2023},
  volume={71},
  number={8},
  pages={6606-6616},
  doi={10.1109/TAP.2023.3283136}}

@article{wolff2023continuous,
  title={Continuous beam steering with a varactor-based reconfigurable intelligent surface in the {Ka-band at 31 GHz}},
  author={Wolff, Alexander and Franke, Lars and Klingel, Steffen and Krieger, Janis and Mueller, Lukas and Stemler, Ralf and Rahm, Marco},
  journal={J. Applied Physics},
  volume={134},
  number={11},
  year={2023},
  publisher={AIP Publishing}
}

@article{yang2020design,
  title={Design and experimental verification of a liquid crystal-based terahertz phase shifter for reconfigurable reflectarrays},
  author={Yang, Jun and Gao, Sheng and Wang, Peng and Yin, Zhiping and Lu, Hongbo and Lai, Weien and Li, Ying and Deng, Guangsheng},
  journal={J. Infrared, Millimeter, and Terahertz Waves},
  volume={41},
  pages={665--674},
  year={2020},
  publisher={Springer}
}

@ARTICLE{Phillips2013,
  author={Phillips, Caleb and Sicker, Douglas and Grunwald, Dirk},
  journal={IEEE Commun. Surveys and Tutorials}, 
  title={A Survey of Wireless Path Loss Prediction and Coverage Mapping Methods}, 
  year={2013},
  volume={15},
  number={1},
  pages={255-270},
  doi={10.1109/SURV.2012.022412.00172}}

@book{kreyszig,
  title={Advanced engineering mathematics},
  author={Erwin Kreyszig and Herbert Kreyszig and Edward J. Norminton},
  year={2011},
  publisher={John Wiley and Sons, Inc.}
}

@INPROCEEDINGS{delbari2024temperature,
  title={Temperature-Aware Phase-shift Design of {LC-RIS} for Secure Communication},
  author={Delbari, Mohamadreza and Wang, Bowu and Gholian, Nairy Moghadas and Asadi, Arash and Jamali, Vahid},
address={Montreal, Canada},
  booktitle={Proc. IEEE International Conf. Commun. (ICC)},
  year={2025},
volume={},
  number={},
  pages={6838-6843},
  doi={10.1109/ICC52391.2025.11161472}}

@inproceedings{delbari2024nearfield,
  author    = {Delbari, Mohamadreza and Alexandropoulos, George C and Schober, Robert and Poor, H Vincent and Jamali, Vahid},
  title     = {Near-Field Multipath {MIMO} Channel Model for Imperfect Surface Reflection},
  booktitle = {Proc. IEEE Global Conf. Commun. Workshops (Globecom Workshops)},
  address={Cape Town, South Africa},
  year      = {2024},
  pages      = {1-7},
  doi       = {10.1109/GCWkshp64532.2024.11101277},
}

@ARTICLE{Liu2023nearfield,
  author={Liu, Yuanwei and Wang, Zhaolin and Xu, Jiaqi and Ouyang, Chongjun and Mu, Xidong and Schober, Robert},
  journal={IEEE Open J. the Commun. Society}, 
  title={Near-Field Communications: A Tutorial Review}, 
  year={2023},
  volume={4},
  number={},
  pages={1999-2049},
  doi={10.1109/OJCOMS.2023.3305583}}

@article{frank1958liquid,
  title={{I.} Liquid crystals. On the theory of liquid crystals},
  author={Frank, Frederick C},
  journal={Discussions of the Faraday Society},
  volume={25},
  pages={19--28},
  year={1958},
    month={february},
  publisher={Royal Society of Chemistry}
}

@book{Selinger2024,
author={Selinger, Jonathan V.},
title={Dynamics and Statistical Mechanics},
bookTitle={Introduction to Topological Defects and Solitons: In Liquid Crystals, Magnets, and Related Materials},
year={2024},
publisher={Springer Nature Switzerland},
pages={41--52},
isbn={978-3-031-70200-6},
doi={10.1007/978-3-031-70200-6_4},
}

@INPROCEEDINGS{Alexandropoulos2022,
  author={Alexandropoulos, George C. and Jamali, Vahid and Schober, Robert and Poor, H. Vincent},
  booktitle={Proc. IEEE Sensor Array and Multichannel Signal Process. Workshop (SAM)}, 
  title={Near-Field Hierarchical Beam Management for {RIS}-Enabled Millimeter Wave Multi-Antenna Systems}, 
  year={2022},
  address={Trondheim, Norway},
  volume={},
  number={},
  pages={460-464},
  doi={10.1109/SAM53842.2022.9827873}}

@book{blinov2012electrooptic,
  title={Electrooptic effects in liquid crystal materials},
  author={Blinov, Lev M and Chigrinov, Vladimir G},
  year={2012},
  publisher={Springer Science \& Business Media}
}

@ARTICLE{Jamali2022lowtozero,
  author={Jamali, Vahid and Alexandropoulos, George C. and Schober, Robert and Poor, H. Vincent},
  journal={IEEE Commun. Lett.}, 
  title={Low-to-Zero-Overhead {IRS} Reconfiguration: Decoupling Illumination and Channel Estimation}, 
  year={2022},
  volume={26},
  number={4},
  pages={932-936},
  doi={10.1109/LCOMM.2022.3141206}}

@article{wright2015coordinate,
  title={Coordinate descent algorithms},
  author={Wright, Stephen J},
  journal={Mathematical programming},
  volume={151},
  number={1},
  pages={3--34},
  year={2015},
  publisher={Springer}
}

@article{sayyah1992anomalous,
  title={Anomalous liquid crystal undershoot effect resulting in a nematic liquid crystal-based spatial light modulator with one millisecond response time},
  author={Sayyah, Keyvan and Wu, Chiung-Sheng and Wu, Shin-Tson and Efron, Uzi},
  journal={Applied physics letters},
  volume={61},
  number={8},
  pages={883--885},
  year={1992},
  publisher={American Institute of Physics}
}

@ARTICLE{Lu2024,
  author={Lu, Yu and others},
  journal={IEEE Trans. Commun.}, 
  title={Performance Analysis of {RIS}-Assisted Communications With Hardware Impairments and Channel Aging}, 
  year={2024},
  volume={72},
  number={6},
  pages={3720-3735},
  doi={10.1109/TCOMM.2024.3361554}}

@ARTICLE{Lu2025,
  author={Lu, Yu and others},
  journal={IEEE Wireless Commun.}, 
  title={Energy-Efficient {RIS}-Aided Cell-Free Massive {MIMO} Systems: Application, Opportunities, and Challenges}, 
  year={2025},
  volume={32},
  number={4},
  pages={148-155},
  doi={10.1109/MWC.001.2400224}}

@ARTICLE{Zeng2024,
  author={Zeng, Shuhao and others},
  journal={IEEE Wireless Commun.}, 
  title={{RIS}-Based {IMT}-2030 Testbed for {MmWave} Multi-Stream Ultra-Massive MIMO Communications}, 
  year={2024},
  volume={31},
  number={3},
  pages={375-382},
  doi={10.1109/MWC.005.2300052}}

@ARTICLE{Zeng2022,
  author={Zeng, Shuhao and others},
  journal={IEEE Trans. Commun.}, 
  title={Intelligent Omni-Surfaces: Reflection-Refraction Circuit Model, Full-Dimensional Beamforming, and System Implementation}, 
  year={2022},
  volume={70},
  number={11},
  pages={7711-7727},
  doi={10.1109/TCOMM.2022.3207804}}

@ARTICLE{Mukherjee2017,
  author={Mukherjee, Sandeep and Das, Suvra Sekhar and Chatterjee, Aritra and Chatterjee, Sourav},
  journal={IEEE Access}, 
  title={Analytical Calculation of {Rician} {K}-Factor for Indoor Wireless Channel Models}, 
  year={2017},
  volume={5},
  number={},
  pages={19194-19212},
  doi={10.1109/ACCESS.2017.2750722}}

@ARTICLE{qingqing2021,
  author={Wu, Qingqing and Zhang, Shuowen and Zheng, Beixiong and You, Changsheng and Zhang, Rui},
  journal={IEEE Trans. Commun.}, 
  title={Intelligent Reflecting Surface-Aided Wireless Communications: A Tutorial}, 
  year={2021},
  volume={69},
  number={5},
  pages={3313-3351},
  doi={10.1109/TCOMM.2021.3051897}}

\begin{IEEEbiography}[{\includegraphics[width=1in,height=1.25in,clip,keepaspectratio]{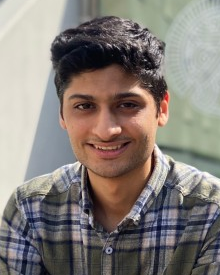}}]{Mohamadreza Delbari}
 (Student Member, IEEE) received the B.Sc. degree in Electrical Engineering from Ferdowsi University, Mashhad, Iran, in 2019, and his M.Sc. degree in Telecommunications from Sharif University of Technology (SUT), Tehran, Iran, in 2022. In 2023, he joined the Resilient Communication Systems Lab, Department of Electrical Engineering and Information Technology, Technical University of Darmstadt (TUDa), where he is currently working as a Researcher. His research interests include Reconfigurable intelligent surfaces (RIS), Near-field communications, and Liquid crystal RIS (LC RIS). He has served as a reviewer for several journals and conferences, including the \textsc{IEEE Transactions on Wireless Communications}, \textsc{IEEE Transactions on Communications}, \textsc{IEEE Wireless Communications Letters}, and the \textsc{IEEE Open Journal of the Communications Society}.
\end{IEEEbiography}
\newpage
\begin{IEEEbiography}[{\includegraphics[width=1in,height=1.25in,clip,keepaspectratio]{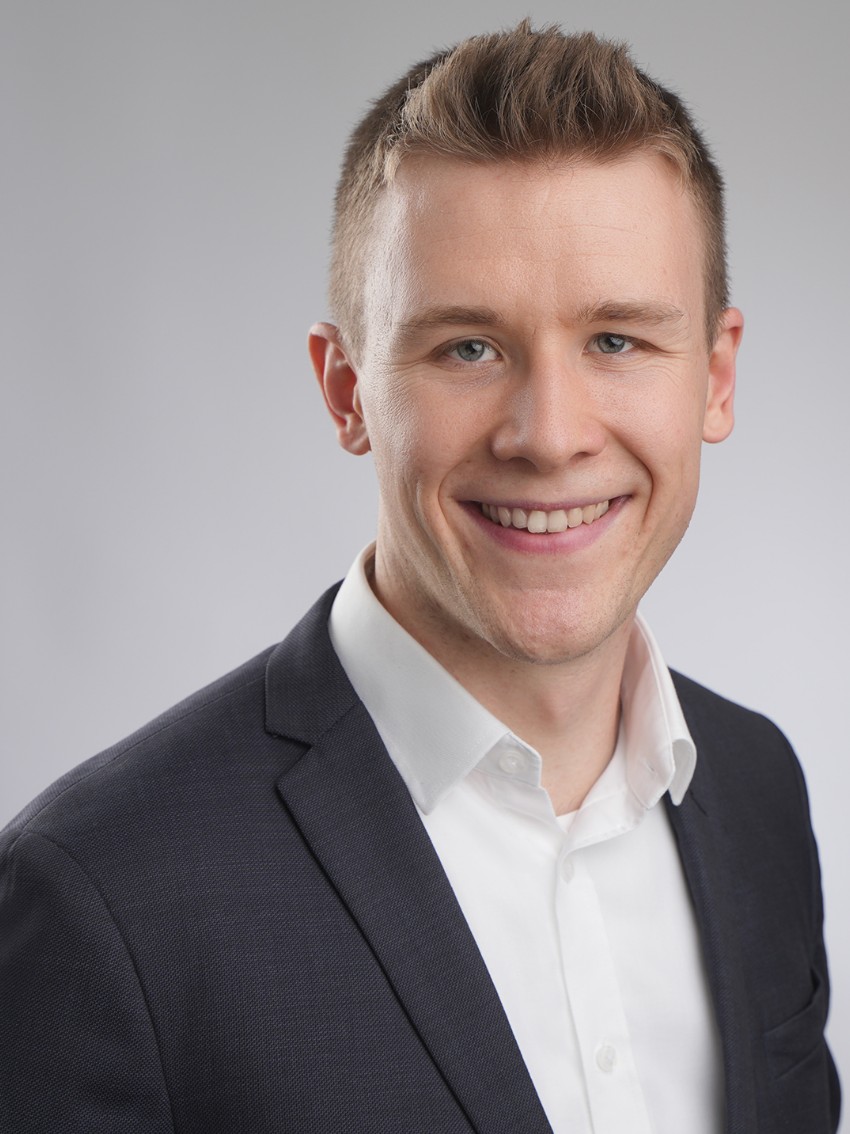}}]{Robin Neuder}
(Student Member, IEEE) received the B.S. and M.S. degree in electrical engineering and information technology from TU Darmstadt, Germany in 2019 and 2021, respectively. Since 2021, he has been a Research Assistant with the Institute of Microwave Engineering and Photonics at TU Darmstadt. His research interests include compact liquid-crystal based tuneable planar devices and their integration into Reconfigurable Intelligent Surfaces.
\end{IEEEbiography}

\begin{IEEEbiography}[{\includegraphics[width=1in,height=1.25in,clip,keepaspectratio]{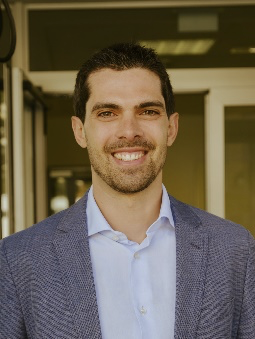}}]{Alejandro Jiménez-Sáez}
    (Member, IEEE) (S’15, M’22) received the master’s degrees (Hons.) in telecommunications engineering from the Polytechnic University of Valencia, Spain, and in electrical engineering from the Technical University of Darmstadt, Germany, in 2017. In 2021, he received the Dr.-Ing. degree (with distinction) in electrical engineering from the TU Darmstadt and the Freunde der TU Darmstadt prize for the best dissertation in electrical engineering. He obtained the Athene Young Investigator award at TU Darmstadt and leads the Smart RF Systems based on Artificial and Functional Materials independent research group. His current research interests include chipless RFID, electromagnetic bandgap structures, liquid crystal, and reconfigurable intelligent surfaces.
\end{IEEEbiography}
\begin{IEEEbiography}[{\includegraphics[width=1in,height=1.25in,clip,keepaspectratio]{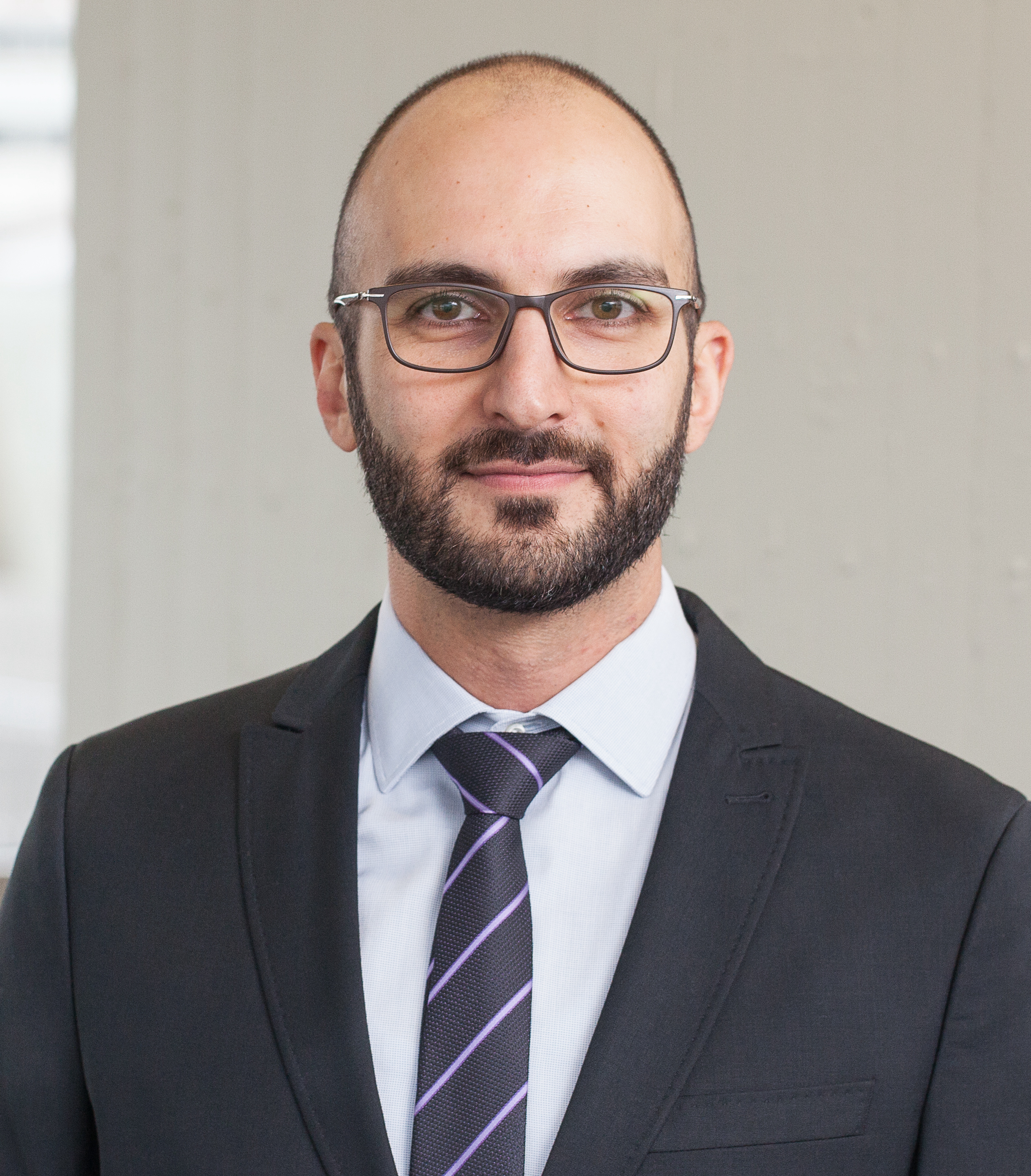}}]{Arash Asadi}
(Senior Member, IEEE) is currently
an Assistant Professor with the Embedded Systems
Group, Delft University of Technology (TU Delft),
Delft, Netherlands, where he leads the Wireless Communication and Sensing Laboratory (WISE). His research interests include wireless communication and
sensing in 6G networks. He was the recipient of several awards, including the Athena Young Investigator
Award from Technische Universität Darmstadt and
Educational Fellow from TU Delft.
\end{IEEEbiography}
\begin{IEEEbiography}[{\includegraphics[width=1in,height=1.25in,clip,keepaspectratio]{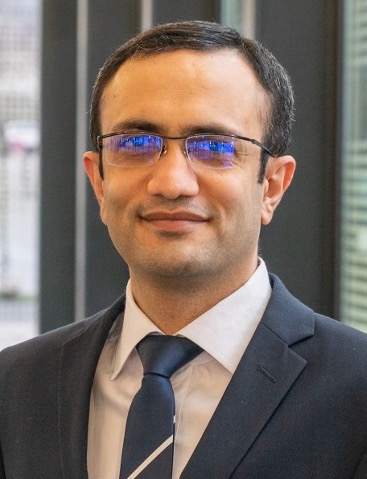}}]{Vahid Jamali}
(Senior Member, IEEE) received the doctoral degree (with distinctions) from Friedrich-Alexander University Erlangen-Nürnberg (FAU) in 2019. He has been an Assistant Professor with the Technical University of Darmstadt (TUDa), since 2022, leading the Resilient Communication Systems Lab. Prior to joining TUDa, he held academic appointments at Princeton University (2021–2022) and FAU (2019– 2021), as a Post-Doctoral Researcher; and at Stanford University as a Visiting Researcher in 2017. His research interests include wireless and molecular communications. He has served as an Associate Editor of the \textsc{IEEE Transactions on Communications} and \textsc{IEEE Communications Letters}, an Editor-at-Large at \textsc{IEEE Open Journal of the Communications Society} as well as a Vice-Chair for the IEEE ComSoc -- German chapter. He has received several awards for his publications including the Best Paper Awards from the IEEE ICC in 2016, the ACM NanoCOM in 2019, the Asilomar CSSC in 2020, the IEEE WCNC in 2021; and the ACM NanoCOM in 2025, the Best Journal Paper Award (Literaturpreis) from the German Information Technology Society (ITG) in 2020, and the Best Paper Award of TAOS Technical Committee of IEEE ComSoc in 2024.
\end{IEEEbiography}


\end{document}